\documentclass[english]{llncs}
\usepackage[T1]{fontenc}
\usepackage[utf8x]{inputenc}
\usepackage{babel}
\usepackage{units}
\usepackage{amsmath}
\usepackage{amssymb}
\usepackage{graphicx}
\usepackage[numbers]{natbib}
\usepackage[unicode=true]
 {hyperref}
\usepackage{breakurl}

\makeatletter
\usepackage[vlined,plain,noresetcount]{algorithm2e} 
\usepackage{lineno}
\SetAlFnt{\small\sf}

\usepackage{graphicx}
\usepackage{xcolor}
\usepackage[adobe-utopia]{mathdesign}

\DeclareMathOperator{\sig}{\emph{sig}}
\DeclareMathOperator{\block}{\emph{block}}
\DeclareMathOperator{\oute}{out}
\DeclareMathOperator{\ine}{in}
\DeclareMathOperator{\id}{id}
\DeclareMathOperator{\smax}{smax}
\DeclareMathOperator{\ms}{ms}
\DeclareMathOperator{\mv}{mv}

\DeclareMathOperator{\sbmv}{\emph{subBlockMap.values}}
\DeclareMathOperator{\new}{\emph{new}}
\DeclareMathOperator{\lab}{\emph{lab}}

\usepackage[margin=1in]{geometry}

\makeatother

\begin{document}

\title{Concurrent bisimulation algorithm}

\author{Konrad Kułakowski }

\institute{Department of Applied Computer Science,\\
AGH University of Science and Technology\\
Al. Mickiewicza 30, \\
30-059 Cracow, Poland\\
\email{konrad.kulakowski@agh.edu.pl}}
\maketitle
\begin{abstract}
The coarsest bisimulation-finding problem plays an important role
in the formal analysis of concurrent systems. For example, solving
this problem allows the behavior of different processes to be compared
or specifications to be verified. Hence, in this paper an efficient
concurrent bisimulation algorithm is presented. It is based on the
sequential \emph{Paige} and \emph{Tarjan} algorithm and the concept
of the state signatures. The original solution follows \emph{Hopcroft's}
principle \emph{``process the smaller half''}. The presented algorithm
uses its generalized version \emph{``process all but the largest
one''} better suited for concurrent and parallel applications. The
running time achieved is comparable with the best known sequential
and concurrent solutions. At the end of the work, the results of tests
carried out are presented. The question of the lower bound for the
running time of the optimal algorithm is also discussed.
\end{abstract}

\section{Introduction}

Modeling of concurrent systems is an important but difficult task.
A variety of methods and formalisms, including \emph{Petri} nets~\citep{Diaz2013pnfm},
different process algebras~\citep{Roscoe2005ttap,Bergstra2006otdr,Kwiatkowska2009ppmc,Balicki2009fdox},
state machines, temporal logic and others, have been developed to
solve this problem. Constructed tools allow users to model and analyze
interprocess communication and interaction \citep{Calzolai2008tatf,Garavel2011cadp,Bezemer2011lhrt,Armstrong2012}.
During analysis, the question of whether two (or more) processes operate
identically very often comes up. One way to answer this question~\citep{vanGlabbeek20013}
leads through checking the relation of bisimulation equivalence between
initial states of the compared processes. The notion of action-based
bisimulation has been brought up independently by \emph{Milner}~\citep{Milner1980acoc}
and \emph{Park}~\citep{Park1981caao}. Over time, it has become the
basis for other equivalence relations, such as branching bisimulation~\citep{vanGlabbeek1996btab}
or \emph{Stutter} bisimulation~\citep{Baier2008pomc}. In practice,
the decision about bisimilarity between initial states needs to solve
the more general \emph{Relational Coarsest Partition Problem} \emph{(RCPP)}.
The first effective algorithm for \emph{RCPP} was given by \emph{Kanellakis}
and \emph{Smolka} \citep{Kanellakis1983cefs}. Their algorithm has
the time complexity $O(|\mathsf{T}|\cdot|\mathsf{S}|)$ where $|\mathsf{S}|$
denotes the number of states, and $|\mathsf{T}|$ is the number of
transitions among them. One year later, the simplified, single function
\emph{RCPP} problem was addressed by \emph{Paige} and \emph{Tarjan}~\citep{Paige1984alta}.
They provided the general solution three years later~\citep{PaigeTarjan1987tpra}.
Achieving an excellent running time $O(|\mathsf{T}|\log|\mathsf{S}|)$
was possible thanks to the adoption of the brilliant algorithmic strategy
\emph{``process the smaller half''} proposed by \emph{Hopcroft}~\citep{Hopcroft1971anln}.
The \emph{Page} and \emph{Tarjan} algorithm inspired other researchers
for further work on bisimulation algorithms~\citep{Fernandez1989,Baier1997wbff}.
An algorithm deciding on branching bisimulation has been proposed
by \emph{Groote} and \emph{Vaandrager}~\citep{Groote1990aeaf}. \emph{RCPP}
was also the subject of research in the field of concurrent and parallel
algorithms. \emph{Balc{\'a}zar} et~al.~\citep{Balcazar1992dbip},
reduced a bisimulation problem to the \emph{Monotone Alternating Circuit
Value Problem,} and thus, showed that deciding on bisimilarity is
\emph{P-complete}. A few years later, two parallel algorithms modeled
on \citep{Kanellakis1983cefs,PaigeTarjan1987tpra}~were shown by
\emph{Rajasekaran} and \emph{Lee}~\citep{Rajasekaran1998pafr}. Their
algorithms achieve $O(|\mathsf{S|^{1+\epsilon})}$ for $\epsilon>0$
and $O(|\mathsf{S|\log|\mathsf{S|)}}$ concurrent running times correspondingly.
\emph{Jeong} et al. \citep{Jeong98afpi} claimed that the \emph{Kanellakis}
and \emph{Smolka} algorithm has parallel implementation running in
$O(|\mathsf{S|})$ time. The parallel algorithm for a single function
\emph{RCPP} is proposed in~\citep{Ha1999aepa}. An efficient distributed
algorithm for strong bisimulation is proposed by \emph{Blom} and\emph{~Orzan}~\citep{Blom2002adaf}.
The algorithm uses the concept of state signatures, which are subsequently
determined during the following partition refinements. Studies on
the use of state signatures were continued in~\citep{Blom2009dbbm}.
An important contribution into the research on \emph{RCPP} is the
bisimulation algorithm given by \emph{Dovier}~et~al.~\citep{Dovier2004aeaf}.
The running time of this algorithm in the case of an acyclic \emph{LTS}
graph is $O(|\mathsf{S}|)$. The works~\citep{PaigeTarjan1987tpra,Dovier2004aeaf}
were the source of inspiration for the incremental bisimulation algorithm
provided by \emph{Saha}~\citep{Saha2007aiba}. 

The algorithm presented in this article uses the notion of a state
signature similar to that defined in \citep{Blom2009dbbm}. However,
in contrast to the solution presented there, it follows the \emph{Page}
and \emph{Tarjan} approach proposed in \citep{PaigeTarjan1987tpra}.
The result is a concurrent algorithm that combines the speed of \emph{Blom
}and\emph{ Orzan's} approach with the robustness of the\emph{ ``process
the smaller half''} strategy. Combining a multi-way splitting together
with \emph{Hopcroft's }strategy in the context of the concurrent processing
results in formulation of the principle - \emph{``process all but
the largest one''}. The presented algorithm intensively uses data
structures such as hash tables, queues and sets. Hence, its final
efficiency highly depends on the effectiveness of these structures.
One of the important performance indicators is the expected (average)
running time. It describes the speed of the algorithm as it usually
is. Therefore, it is very important in practice. The analysis carried
out in this work focuses on the expected running time of the presented
algorithm. On the sequential case there is $O(|\mathsf{T}|\log|\mathsf{S}|)$.
Hence, it is as good as the best known solution \citep{PaigeTarjan1987tpra}.
The worst case scenario running time of the algorithm's concurrent
implementation tends to $O(|\mathsf{S}|\log\beta)$, where $\beta$
is the maximal number of transitions outgoing from a single state.
Thus, assuming that between two different states there could be at
most one transition, the presented solution achieves a time complexity
known from \citep{Rajasekaran1998pafr}. The second estimation, however,
is computed on the assumption that all the hash maps are implemented
in the form of directly addressed arrays \citep{Cormen2009ita}. This
ensures full parallelism, although it often leads to a high demand
for resources. Therefore, in practice, the concurrent implementations
must use concurrent data structures \citep{Moir2005hods}, which provide
a reduced degree of parallelism. Hence, the estimation $O(|\mathsf{S}|\log\beta)$
needs to be treated as the lower bound for the expected running time
for the concurrent implementations of this algorithm. Nevertheless,
the fully parallel implementation is possible and, for some specific
\emph{LTS} graphs, it may be useful. A detailed analysis of the fully
parallel algorithm is presented in (Sec. \ref{sub:Concurrent-complexity}).

The article is composed of eight sections, where the first two provide
a brief literature review, and introduce indispensable theoretical
definitions. Section \ref{sec:Bisimulation-Algorithm} describes the
general, sequential form of the algorithm. All the procedures are
defined and described there. The principle \emph{``process all but
the largest one'' }is separately explained in (Sec. \ref{sub:Process-all-but-the-largest-one}).
The next section (Sec. \ref{sec:Sequential-complexity}) analyses
the sequential running time of the algorithm. The concurrent solution
is discussed in Section \ref{sub:Concurrent-algorithm}. The matter
of its optimality is addressed in Section \ref{sec:Notes-on-the-solution-optimality}.
The penultimate section (Sec. \ref{sub:Concurrent-implementation})
discusses the issues related to concurrent implementation and provides
preliminary test results. The work ends with Section \ref{sec:Summary}
- \emph{Summary}.

\section{Preliminary information\label{sec:Preliminaries}}

In this section, the notion of bisimulation equivalence (referred
further as bisimulation) and the necessary definitions are introduced.
The most fundamental concept used in the context of bisimulation is
a labelled transition system \emph{(LTS)} representing all the possible
states and transitions of a model. 
\begin{definition}
Let LTS $ $be a triple $(\mathsf{S},\mathsf{T},L)$ where $\mathsf{S}$
is a finite set of states, $\mathsf{T}\subset S\times S$ is the finite
set of transitions, and $L:S\times S\rightarrow A$ is the labeling
function over the finite alphabet $A$. The situation when $L(u,v)=a$
will also be denoted as $u\overset{a}{\rightarrow}v$.
\end{definition}
From the user's perspective, it is convenient to consider \emph{LTS}
as a directed labelled graph where vertices represent states and transitions
represent edges. Bearing in mind this analogy, the terms such as states,
vertices, transition and edges are used interchangeably. 

A signature of $u\in S$ is formed by the set of labels of all the
edges starting from $u$, i.e. $\sig(u)\overset{df}{=}\{L(u,v)|(u,v)\in\mathsf{T}\}$.
Similarly, the signature of a group of states is formed by the sum
of state signatures i.e. $\sig(P)\overset{df}{=}\bigcup_{u\in P}\sig(u)$
where $P\subset S$. The set of directly reachable neighbors of $u$
in \emph{LTS} is called output states of $u$ and denoted as $\oute(u)\overset{df}{=}\{v|(u,v)\in\mathsf{T}\}$.
The state $u$ is directly reachable from elements of $\ine(u)\overset{df}{=}\{v|(v,u)\in\mathsf{T}\}$
called input states of $u$. Similarly, the input and output states
of $Q\subset S$ are defined as $\ine(Q)\overset{df}{=}\bigcup_{q\in Q}\ine(q)$
and $\oute(Q)\overset{df}{=}\bigcup_{q\in Q}\oute(q)$.  
\begin{definition}
\label{bisim-def}Let $(\mathsf{S},\mathsf{T},L)$ be a labelled transition
system. A bisimulation between elements of \textup{$\mathsf{S}$}
is the relation $\sim\subseteq\mathsf{S}\times\mathsf{S}$ so that
whenever $(s_{1},s_{2})\in\sim$ then $\forall a\in A$ holds
\begin{itemize}
\item if $L(s_{1},s_{1}^{'})=a$ then there is an $s_{2}^{'}$ such that
$L(s_{2},s_{2}^{'})=a$ and $(s_{1}^{'},s_{2}^{'})\in\sim$
\item if $L(s_{2},s_{2}^{'})=a$ then there is an $s_{1}^{'}$ such that
$L(s_{1},s_{1}^{'})=a$ and $(s_{1}^{'},s_{2}^{'})\in\sim$
\end{itemize}
\end{definition}
The two states $s_{1},s_{2}\in\mathsf{S}$ are bisimulation equivalent
(or bisimilar), written $s_{1}\sim s_{2}$, if there exists a bisimulation
$\sim$ so that $(s_{1},s_{2})\in\sim$. Every equivalence relation
in the set of states $\mathsf{S}$ defines a partition $\mathcal{P}\subset2^{\mathsf{S}}$
into non-empty disjoint subsets (equivalence classes). Conversely,
every partition $\mathcal{P}\subset2^{\mathsf{S}}$ defines an equivalence
relation in $\mathsf{S}$ (see \emph{Decomposition Theorem} \citep[p. 297]{Bronstein2005tdm}).
Since a bisimulation is an equivalence relation, thus their equivalence
classes in a natural way determine the division of $S$ into subsets.
In particular, two elements $s_{1},s_{2}\in\mathsf{S}$ are bisimilar
if they belong to the same equivalence class of some bisimulation.
According to the \emph{Decomposition Theorem}, a bisimulation equivalence
is uniquely represented by some partition $\mathcal{P}\subset2^{\mathsf{S}}$.
Therefore, to decide whether $s_{1}\sim s_{2}$, first, $\mathcal{P}$
needs to be computed, then checked if $s_{1},s_{2}\in P$ for some
$P\in\mathcal{P}$. The presented algorithm focuses on the concurrent
calculation of $\mathcal{P}$, and assumes that the partition membership
problem for $s_{1}$ and $s_{2}$ is simple and can be efficiently
verified.

\section{Sequential algorithm\label{sec:Bisimulation-Algorithm}}

During its sequential execution, the algorithm uses a few global data
structures. These are: $\mathcal{P}$ - set of partitions (as explained
later, it is used mainly for demonstration purposes) and \emph{initPartition}
- set of partitions after the initialization phase. Every partition
(referred also as a block) has its own identifier. The mapping between
identifiers and blocks is stored in the \emph{blockById} linked map.
Similarly, each state belongs to the block with the specified identifier.
The mapping between states and blocks, applicable in the current step
of the algorithm, is stored in the linked map \emph{stateToBlockId}.
The new mapping, for the next step of the algorithm, is stored in
the \emph{nextStateToBlockId} linked map. The queues $\mathcal{M}$
and $\mathcal{S}$ hold blocks to be marked and to be used as splitters
respectively (Listing \ref{listing:bisimalg:structures}). 

\IncMargin{1em}  
\LinesNumbered  
\begin{algorithm}[h]

\texttt{set $\mathcal{P}$}

\texttt{linked map $ $}\emph{initPartition}

\texttt{linked map }\emph{blockById}

\texttt{linked map }\emph{stateToBlockId}

\texttt{linked map }\emph{nextStateToBlockId}

\texttt{queue $\mathcal{M}$}

\texttt{queue $\mathcal{S}$}

\NoCaptionOfAlgo
\medskip \caption{\textbf{Listing \thealgocf: } Bisimulation algorithm - the global data structures}

\label{listing:bisimalg:structures}

\end{algorithm}
\DecMargin{1em} 

The main routine of the algorithm consists of two parts. The first,
initialization, is responsible for preparing an initial version of
the partition $\mathcal{P}$ and filling the auxiliary data structures
(Listing: \ref{listing:bisimalg:main}, line: \ref{code:bisimalg:main:initialization}).
The second one consists of three cyclic steps: mark - determines blocks
that need to be refined, split - performs block splitting, and copy
- updates auxiliary data structures according to the performed refinement
(Listing: \ref{listing:bisimalg:main}, line: \ref{code:bisimalg:main:mark-split-copy}).
All the three mark, split and copy steps of the main part of the algorithm
are repeated as long as further refinements are required (i.e. queues
$\mathcal{M}$ is not empty). 

\IncMargin{1em}  
\LinesNumbered  
\begin{algorithm}[h]

\texttt{BisimulationAlgorithm()}

\texttt{~~InitializationPhase()\label{code:bisimalg:main:initialization}}

\texttt{~~if $|$$\mathcal{P}$$|\neq1$ \label{code:bisimalg:main:processing_cond}}

\texttt{~~~~MarkSplitCopyPhase()\label{code:bisimalg:main:mark-split-copy}}

\texttt{~~return $\mathcal{P}$}

\NoCaptionOfAlgo
\medskip \caption{\textbf{Listing \thealgocf: } Bisimulation algorithm - main routine}

\label{listing:bisimalg:main}

\end{algorithm}
\DecMargin{1em}

\subsection{Initialization Phase }

The purpose of the initialization phase is to prepare the first, initial
version of the partition $\mathcal{P}$ (and \emph{stateToBlockId}
map). This phase is composed of the three subroutines that are responsible
for grouping states in $\mathsf{S}$ into blocks so that every two
states from the same block have an identical set of labels of the
outgoing edges. Thus, after the initialization phase for every $P\in\mathcal{P}$,
and for all $u,v\in P$, it holds that $sig(u)=sig(v)$. Of course,
if all the states have the same signatures, the initial version of
$\mathcal{P}$ contains only one block. In such a case it is easy
to prove that every two states satisfy the bisimulation relation,
thus, no further calculations are needed. Thus, after confirming that
the cardinality of $\mathcal{P}$ is one (the number of different
keys in \emph{blockById} is one) the algorithm ends up returning $\mathcal{P}$
on the output (Listing: \ref{listing:bisimalg:main}, Line: \ref{code:bisimalg:main:processing_cond}).
Otherwise the initialization phase must be followed by the mark-split-copy
phase as discussed in (Sec. \ref{sub:Mark-Split-Copy-Phase}). 

\IncMargin{1em}  
\LinesNumbered  
\begin{algorithm}[h]

\texttt{InitializationPhase()}

\texttt{~~parallel for $s\in\mathsf{S}$ \label{code:initializationPhase:StateSignatureInit-start}}

\texttt{~~~~}\emph{StateSignatureInit($s$)}\texttt{\label{code:initializationPhase:StateSignatureInit-end}}

\texttt{~~parallel for $s\in\mathsf{S}$ \label{code:initializationPhase:PartInitStart}}

\texttt{~~~~}\emph{PartitionInit($s$)}\texttt{\label{code:initializationPhase:PartInitEnd}}

\texttt{~~}\emph{let us denote}\texttt{ $\{(\sig_{1},\block_{1}),\ldots,$ }

\texttt{~~~~$(\sig_{r},\block_{r})\}=$}\emph{ initPartition}

\texttt{~~parallel for $(\sig_{i},\block_{i})\in$ }\emph{initPartition}\texttt{\label{code:initializationPhase:AuxStructInit-iter-begin}}

\texttt{~~~~}\emph{AuxStructInit(}\texttt{$\block_{i}$}\emph{)}\texttt{\label{code:initializationPhase:AuxStructInit-iter-end}}

\NoCaptionOfAlgo
\medskip \caption{\textbf{Listing \thealgocf: } Bisimulation algorithm - Initialization phase}

\label{listing:bisimalg:initphase}

\end{algorithm}
\DecMargin{1em} 

The main initialization procedure is \emph{InitializationPhase} (Listing:
\ref{listing:bisimalg:initphase}). It splits the set $\mathsf{S}$
of states into $k$ subsets ($k$ - is determined by the number of
available processors), then processes them successively using \emph{StateSignatureInit}(),
\emph{PartitionInit()} and \emph{AuxStructInit}(). Splitting $\mathsf{S}$
into possibly equal subsets is desirable for the parallel processing
performance. It contributes to the even distribution of computing,
however, as will be discussed later on, it does not guarantee running
time reduction. To achieve the desired level of parallelism, further
code parallelization is needed. For the purpose of understanding the
idea of the sequential algorithm, it is enough to treat each \emph{parallel
for} instruction as a simple sequential iteration. 

The first sub-procedure \emph{StateSignatureInit} (Listing: \ref{listing:bisimalg:initialization-threephases},
Lines: \ref{code:initialization:VertexSigInit-begin} - \ref{code:initialization:VertexSigInit-end})
is responsible for creating state signatures (Listing: \ref{listing:bisimalg:initialization-threephases},
line: \ref{code:initialization:VertexSigInit-sig-creation}). The
signatures are used to index new blocks in the \emph{initPartition()}
map (Listing: \ref{listing:bisimalg:initialization-threephases},
Lines: \ref{code:initialization:VertexSigInit-contains-sig} - \ref{code:initialization:VertexSigInit-end}).
Thus, the \emph{StateSignatureInit()} creates mapping between state
signatures and newly (empty) created blocks. The actual states are
assigned to these blocks in the next sub-procedure. Placing the block
creation and state assignment into two different sub-procedures simplifies
synchronization of the \emph{initPartition} map. In particular, there
is no need to block the whole \emph{initPartition} map, which would
be indispensable if both operations (the new block creation and adding
new states to them) had been implemented within the same loop. 

The second sub-procedure \emph{PartitionInit()} scans all states in
$\mathsf{S}$ and assigns them to the previously created blocks (Listing:
\ref{listing:bisimalg:initialization-threephases}, Lines: \ref{code:initialization:PartitionInit-get-sig}
- \ref{code:initialization:PartitionInit-block-update}). Thus, after
the execution of the \emph{PartitionInit()} procedure, all the states
with the same signature $s$ are kept in the block \emph{initPartition.get(s)}.
At the end of the sub-procedure the \emph{stateToBlockId} auxiliary
structure is updated (Listing: \ref{listing:bisimalg:initialization-threephases},
line: \ref{code:initialization:PartitionInit-end}). Next, \emph{stateToBlockId
}is used to easily calculate $\sig(v)$. 

\IncMargin{1em}  
\LinesNumbered  
\begin{algorithm}[h]

\texttt{StateSignatureInit($v$)\label{code:initialization:VertexSigInit-begin}}

\texttt{~~$\sig(v)\leftarrow\{L(v,u):(v,u)\in\mathsf{T}\}$\label{code:initialization:VertexSigInit-sig-creation}}

\texttt{~~if $\neg$}\emph{initPartition.contains}\texttt{$(\sig(v))$\label{code:initialization:VertexSigInit-contains-sig}}

\texttt{~~~~}\emph{initPartition.put($\sig(v)$,newBlock())}\texttt{
\label{code:initialization:VertexSigInit-end}}

\texttt{PartitionInit($v$)\label{code:initialization:PartitionInit-begin}}

\texttt{~~$\block$ $\leftarrow$ }\emph{initPartition.get($\sig(v)$)}\texttt{\label{code:initialization:PartitionInit-get-sig}}

\texttt{~~$\block\leftarrow\block\cup\{v\}$\label{code:initialization:PartitionInit-block-update}}

\texttt{~~}\emph{stateToBlockId.put(v.id,block.id)}\texttt{\label{code:initialization:PartitionInit-end}}

\texttt{AuxStructInit($\block$)\label{code:initialization:AuxStructInit-begin}}

\texttt{~~~~}\emph{blockById.put(block.id,block)}\texttt{\label{code:initialization:BlockById}}

\texttt{~~~~$\mathcal{P}\leftarrow\mathcal{P}\cup\{\block\}$\label{code:initialization:PartitionInitSet}}

\texttt{~~~~$\mathcal{S}\leftarrow\mathcal{S}\cup\{\block\}$\label{code:initialization:AuxStructInit-end}\label{code:initialization:SplitterInitSet}}

\NoCaptionOfAlgo
\medskip \caption{\textbf{Listing \thealgocf: } Bisimulation algorithm - Initialization phase subroutines}

\label{listing:bisimalg:initialization-threephases}

\end{algorithm}
\DecMargin{1em} 

The last sub-procedure \emph{AuxStructInit} creates mapping between
block ids and block references. The mapping is stored in \emph{blockById}
map. It also initially populates the set of splitters $\mathcal{S}$,
and the partition set $\mathcal{P}$.

\subsection{Mark-Split-Copy Phase\label{sub:Mark-Split-Copy-Phase}}

The starting point for the \emph{Mark-Split-Copy} phase is the procedure
\emph{MarkSplitCopyPhase()} (Listing: \ref{listing:bisimalg:marksplitcopyphase}).
It repeats in the sequential loop \emph{while} (Listing: \ref{listing:bisimalg:marksplitcopyphase},
Lines: \ref{code:MarkSplitCopyPhase:mainLoop:Begin} - \ref{code:MarkSplitCopyPhase:mainLoop:End})
three consecutive subroutines: Marking\emph{()}, \emph{Splitting()},
and \emph{Copying()}. Each subroutine is executed concurrently. The
input to the \emph{Marking()} instances are splitters - the blocks,
which induce further partition refinement. 

\IncMargin{1em}  
\LinesNumbered  
\begin{algorithm}[h]

\texttt{MarkSplitCopyPhase()}

\texttt{~~while(true)\label{code:MarkSplitCopyPhase:mainLoop:Begin}}

\texttt{~~~~parallel for $S\in\mathcal{S}$ \label{code:MarkSplitCopyPhase:mainLoop:confor-call-1}}

\texttt{~~~~~~}\emph{Marking($S$)}\texttt{\label{code:MarkSplitCopyPhase:mainLoop:MarkingCall}}

\texttt{~~~~if $\mathcal{M}=\varnothing$ \label{code:MarkSplitCopyPhase:mainLoop:NoMarkedBlocks_cond}}

\texttt{~~~~~~break\label{code:MarkSplitCopyPhase:mainLoop:NoMarkedBlocks_break}}

\texttt{~~~~parallel for $M\in\mathcal{M}$ \label{code:MarkSplitCopyPhase:mainLoop:confor-call-2}}

\texttt{~~~~~~}\emph{Splitting($M$)}\texttt{\label{code:MarkSplitCopyPhase:mainLoop:confor-splittingCall}}

\texttt{~~~~parallel for $(\mbox{s},\id)\in$ }\emph{nextStateToBlockId}\texttt{\label{code:MarkSplitCopyPhase:mainLoop:confor-call-3}}

\texttt{~~~~~~}\emph{Copying(}\texttt{$\mbox{s},\id$}\emph{)}\texttt{\label{code:MarkSplitCopyPhase:mainLoop:End}}

\NoCaptionOfAlgo
\medskip \caption{\textbf{Listing \thealgocf: } Bisimulation algorithm - MarkSplitCopy phase}

\label{listing:bisimalg:marksplitcopyphase}

\end{algorithm}
\DecMargin{1em} 

The initial set of splitters $\mathcal{S}$ is formed by the first
refinement of $\mathcal{P}$ (Listing: \ref{listing:bisimalg:initialization-threephases},
Lines: \ref{code:initialization:PartitionInitSet},\ref{code:initialization:SplitterInitSet}).
Each call \emph{Marking()} takes one splitter from $\mathcal{S}$,
and processes it (Listing: \ref{listing:bisimalg:marksplitcopyphase:marking}).
The result of the \emph{Marking()} loop is $\mathcal{M}$ - the set
of blocks marked to be split. If there is no block in $\mathcal{M}$
the algorithm ends (Listing: \ref{listing:bisimalg:marksplitcopyphase},
Lines: \ref{code:MarkSplitCopyPhase:mainLoop:NoMarkedBlocks_cond}
- \ref{code:MarkSplitCopyPhase:mainLoop:NoMarkedBlocks_break}). Blocks
in $\mathcal{M}$ are also the input to the \emph{Splitting()} subroutine.
Every instance of \emph{Splitting()} takes one block from $\mathcal{M}$.
As a result of \emph{Splitting() }the set $\mathcal{P}$ is refined
(blocks previously indicated in $\mathcal{M}$ are split) and some
of them go to $\mathcal{S}$ as splitters. The modified structure
of $\mathcal{P}$ is also reflected in \emph{nextStateToBlockId}.
Thus,\emph{ }after the \emph{Splitting()} loop execution, the \emph{Copying()}
replicates the newly created mapping \emph{nextStateToBlockId} to
the \emph{stateToBlockId }(Listing: \emph{\ref{listing:bisimalg:marksplitcopyphase:copyWorker}). }

\IncMargin{1em}  
\LinesNumbered  
\begin{algorithm}[h]

\texttt{Copying($\mbox{s},\id$)}

\texttt{~~~~}\emph{stateToBlockId.put($\mbox{s},\id$)}

\NoCaptionOfAlgo
\medskip \caption{\textbf{Listing \thealgocf: } Bisimulation algorithm - copying routine}

\label{listing:bisimalg:marksplitcopyphase:copyWorker}

\end{algorithm}
\DecMargin{1em}

\subsubsection{Marking}

After the initialization phase, the partition $\mathcal{P}$ contains
blocks consisting of states with unique signatures. Thus, for all
$P\in\mathcal{P}$ and for every two $u,v\in P$ it holds that $\sig(u)=\{L(u,t)|(u,t)\in\mathsf{T}\}=\{L(v,r)|(v,r)\in\mathsf{T}\}=\sig(v)$.
The fact that all the states within the same block have identical
signatures allows us to reformulate the condition of bisimulation. 
\begin{theorem}
\label{prop-bisimulation-cond}For every two blocks $P,S\in\mbox{\ensuremath{\mathcal{P}} }$
the states $u,v\in P$ are bisimilar if and only if it holds that
\begin{equation}
\left((u,t),(v,r)\in\mathsf{T}\wedge L(u,t)=L(v,r)=a\wedge t\in S\right)\Rightarrow r\in S\label{eq:bcond}
\end{equation}
\end{theorem}
\begin{proof}
``$\Rightarrow$'' since $u\sim v$ and $(u,t),(v,r)\in\mathsf{T}\wedge L(u,t)=a$
then due to (Def. \ref{bisim-def}) there must be such $r\in\mathsf{S}$
that $L(v,r)=a$ and $t\sim r$. Bisimilarity of $t$ and $r$ implies
that they have the same signatures i.e. $\sig(t)=\sig(r)$. Hence,
due to the construction of $\mathcal{P}$ they belong to the same
block, i.e. $t\in S\Rightarrow r\in S$. 

``$\Leftarrow$'' in contradiction, let us assume that for every
$P,S\in\mbox{\ensuremath{\mathcal{P}}}$ and $u,v\in P$ the right
side of the proposition \ref{prop-bisimulation-cond} is true but
$u$ and $v$ are not bisimilar. It is easy to see that the lack of
bisimilarity between $u$ and $v$ implies that there is a sequence
of symbols $a_{1},a_{2},\ldots,a_{k}$ from $A$ leading to the states
$t_{k},r_{k}$ (i.e. $u\overset{a_{1}}{\rightarrow}t_{1}\overset{a_{2}}{\rightarrow}t_{2}\overset{a_{3}}{\rightarrow}\ldots\overset{a_{k}}{\rightarrow}t_{k}$
and $u\overset{a_{1}}{\rightarrow}r_{1}\overset{a_{2}}{\rightarrow}r_{2}\overset{a_{3}}{\rightarrow}\ldots\overset{a_{k}}{\rightarrow}r_{k}$)
such that $b\in A$ and $t_{k+1}\in\mathsf{S}$ that $t_{k}\overset{b}{\rightarrow}t_{k+1}\in\mathsf{T}$
but there is no $r_{k+1}\in\mathsf{S}$ satisfying $r_{k}\overset{b}{\rightarrow}r_{k+1}\in\mathsf{T}$
(or reversely there is $c\in A$ and $r_{k+1}\in\mathsf{S}$ such
that $r_{k}\overset{c}{\rightarrow}r_{k+1}\in\mathsf{T}$ but there
is no $t_{k+1}\in\mathsf{S}$ so that $t_{k}\overset{c}{\rightarrow}t_{k+1}\in\mathsf{T}$).
In other words, $\sig(t_{k})\neq\sig(r_{k})$. 

However, according to the right side of (Th. \ref{prop-bisimulation-cond})
the assertion $u,v\in P$ implies that there is such $S_{1}\in\mathcal{P}$
that $t_{1},r_{1}\in S_{1}$. Similarly there is $S_{2},\ldots,S_{k}$
such that $t_{2},r_{2}\in S_{2},\ldots,t_{k},r_{k}\in S_{k}$. Hence,
due to the construction of $\mathcal{P}$ it holds that $\sig(t_{k})=\sig(r_{k})$.
Contradiction. 
\end{proof}
The aim of the presented algorithm is to prepare such $\mathcal{P}$
that satisfy the condition (Eq. \ref{eq:bcond}). Of course, it is
possible that for some intermediate partition refinement (Eq. \ref{eq:bcond})
is not true. In such a case, problematic blocks need to be split into
two or more smaller blocks. The search for the candidates to be split,
similarly as in \citep{PaigeTarjan1987tpra}, is slightly contrary
to the natural direction of the bisimulation definition. Thus, first
the block $S$ called splitter is taken from $\mathcal{S}$, then
all its predecessors (blocks $P$ such that $\oute(P)\cap S\neq\varnothing$)
are examined (Listing: \ref{listing:bisimalg:marksplitcopyphase:marking},
line: \ref{bisimalg:marksplitcopyphase:markerWorker:fetching_P}).
When $P\in\mathcal{P}$, for which (Eq. \ref{eq:bcond}) does not
hold, is identified, all the states which may form a new block are
marked. 

Let us assume that $S$ (splitter) is the subject of processing within
the subroutine \emph{Marking()} (Listing: \ref{listing:bisimalg:marksplitcopyphase:marking}).
Let $P$ be the currently examined block, and $a\in\sig(P)$ be the
label so that there exists $(u,t)\in\mathsf{T}\wedge u\in P\wedge t\in S$
and $L(u,t)=a$. Let us define a set of $a$-predecessors of $S$
with respect to $P$ as 
\begin{equation}
\overline{P}(a,S)\overset{df}{=}\{u\in P|\exists(u,t)\in\mathsf{T}\wedge t\in S\wedge a=L(u,t)\}\label{eq:a-pred-of-S-with-respect-to-P-def}
\end{equation}

Note that, if for every $S\in\mathcal{P}$ it holds that $\overline{P}(a,S)=P$
for all $a\in\sig(P)$ then the condition (Equation: \ref{eq:bcond})
holds. If for some intermediate partition refinement $\mathcal{P}$
there is $S,P\in\mathcal{P}$ such that $\overline{P}(a,S)\subsetneq P$
and $a\in\sig(P)$, then the identical signatures of vertices in $P$
imply that there is the vertex $v\in P\backslash\overline{P}(a,S)$
and transition $(v,r)\in\mathsf{T}$ labelled by $a$ such that $r\notin S$
(Fig. \ref{fig:markingPhase}). Then the condition (Equation: \ref{eq:bcond})
is violated. In order to restore the condition (Equation: \ref{eq:bcond})
$P$ needs to be split into $P\backslash\overline{P}(a,S)$ and $\overline{P}(a,S)$.
Therefore, the subroutine \emph{Marking()} first identifies such $P$
block as requiring division (Listing: \ref{listing:bisimalg:marksplitcopyphase:marking},
line: \ref{bisimalg:marksplitcopyphase:markerWorker:split_need_ident}
and \ref{bisimalg:marksplitcopyphase:markerWorker:foreach-splits-keys-end}),
then marks all the states from $\overline{P}(a,R)$ (Listing: \ref{listing:bisimalg:marksplitcopyphase:marking},
line: \ref{bisimalg:marksplitcopyphase:markerWorker:states-mark}). 

\begin{figure}
\begin{centering}
\includegraphics[scale=0.9]{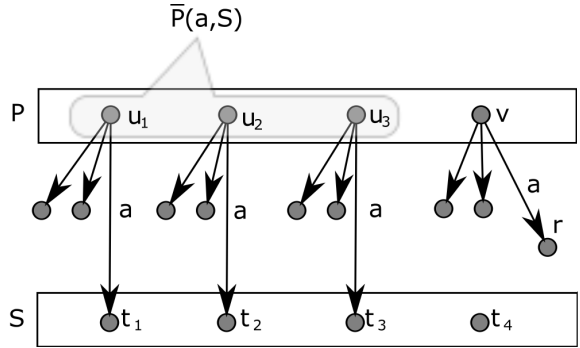}
\par\end{centering}

\caption{Marking step, where $S$ is a splitter, and $P$ violates condition~(Eq.~\ref{eq:bcond}).}

\label{fig:markingPhase}
\end{figure}

Following \citet{PaigeTarjan1987tpra} the block $P$ is said to be
\emph{a-stable} with respect to $S$ if there exists $a\in\sig(P)$
such that $\overline{P}(a,S)=P$, and it is said to be \emph{stable}
with respect to $S$ if it is \emph{a-stable} with respect to $S$
for all $a\in\sig(P)$. Similarly, the partition $\mathcal{P}$ is
said to be stable with respect to the block $S$ if all $P\in\mathcal{P}$
are stable with respect to $S$, and the partition $\mathcal{P}$
is stable with respect to the partition $\mathcal{O}$ if $\mathcal{P}$
is stable with respect to every $S\in\mathcal{O}$. 

Definition (Eq. \ref{eq:a-pred-of-S-with-respect-to-P-def}) implies
that $\overline{P}(a,R)\subseteq P$. Thus, to decide about \emph{a-stability,}
it is enough to compare the cardinality of $\overline{P}(a,R)$ and
$P$ (Listing: \ref{listing:bisimalg:marksplitcopyphase:marking},
line: \ref{bisimalg:marksplitcopyphase:markerWorker:split_need_ident}).
If $|P|\neq|\overline{P}(a,R)|$, or more precisely $|P|>|\overline{P}(a,R)|$,
then $P$ needs to be split, and $\overline{P}(a,R)$ is a potential
candidate for the new block. 

\IncMargin{1em}  
\LinesNumbered  
\begin{algorithm}[h]

\texttt{Marking($S$)}

\texttt{~~}\emph{splitsMap}\texttt{ $\leftarrow\varnothing$ \label{bisimalg:marksplitcopyphase:markerWorker:splitMapDef}}

\texttt{~~parallel for $s\in S$ \label{bisimalg:marksplitcopyphase:markerWorker:foreachv-in-S}}

\texttt{~~~~parallel for $u\in\ine(s)$ \label{bisimalg:marksplitcopyphase:markerWorker:foreachS:begin}}

\texttt{~~~~~~}\emph{pm}\texttt{ $\leftarrow(L(u,v),$ }\emph{stateToBlockId}(\emph{$\id(u)$}))\texttt{
\label{bisimalg:marksplitcopyphase:markerWorker:foreachS:partMarker}}

\texttt{~~~~~~}\emph{splitsMap.updateValueSet(pm, u)}\texttt{\label{bisimalg:marksplitcopyphase:markerWorker:foreachS:stop}}

\texttt{~~parallel for $pm\in$ }\emph{splitsMap.keys()}\texttt{\label{bisimalg:marksplitcopyphase:markerWorker:foreach-splits-keys-begin}}

\texttt{~~~~$B\leftarrow$ }\emph{blockById.get(second(pm))}\texttt{\label{bisimalg:marksplitcopyphase:markerWorker:fetching_P}}

\texttt{~~~~if $|B|>1$ and $|B|>$ $|$}\emph{splitsMap.get(pm)}\texttt{$|$\label{bisimalg:marksplitcopyphase:markerWorker:split_need_ident}}

\texttt{~~~~~~parallel for $u\in$ }\emph{splitsMap.get(pm)}\texttt{
\label{bisimalg:marksplitcopyphase:markerWorker:splits-map-iter}}

\texttt{~~~~~~~~}\emph{mark(}\texttt{$u,B$}\emph{)}\texttt{\label{bisimalg:marksplitcopyphase:markerWorker:states-mark}}

\texttt{~~~~~~$\mathcal{M}\leftarrow\mathcal{M}\cup\{B\}$\label{bisimalg:marksplitcopyphase:markerWorker:foreach-splits-keys-end}}

\NoCaptionOfAlgo
\medskip \caption{\textbf{Listing \thealgocf: } Bisimulation algorithm - marking routine}

\label{listing:bisimalg:marksplitcopyphase:marking}

\end{algorithm}
\DecMargin{1em} 

The marking loop (Listing: \ref{listing:bisimalg:marksplitcopyphase},
Lines: \ref{code:MarkSplitCopyPhase:mainLoop:confor-call-1}-\ref{code:MarkSplitCopyPhase:mainLoop:MarkingCall})
assumes that the set of splitters $\mathcal{S}$ is previously known
(after the initialization phase $\mathcal{S}=\mathcal{P}$). Hence,
it takes one element $S$ from $\mathcal{S}$ (Listing: \ref{listing:bisimalg:marksplitcopyphase},
line: \ref{code:MarkSplitCopyPhase:mainLoop:MarkingCall}) and processes
it inside the \emph{Marking() }procedure. \emph{Marking()} starts
from calculating all the $a$-predecessors of $S$. For this purpose,
it traverses all the incoming edges into S (Listing: \ref{listing:bisimalg:marksplitcopyphase:marking},
Lines: \ref{bisimalg:marksplitcopyphase:markerWorker:foreachv-in-S},
\ref{bisimalg:marksplitcopyphase:markerWorker:foreachS:begin}) and
creates and fills \emph{splitsMap} - the auxiliary concurrent hash
map. The map as a key takes the pair in the form%
\footnote{For the purpose of the algorithm it is assumed that every block is
identified by its unique, integer id $id(P)$. The mappings between
the block and id are provided by the auxiliary maps: \emph{stateToBlockId}
and \emph{blockById}.%
} $(a,\id(P))$, which corresponds to the value $\overline{P}(a,S)$.
Iterating through the keys of \emph{splitsMap}%
\footnote{To facilitate iteration through the key's set in a map, it is useful
to keep all the keys in a separate linked list. For instance, \emph{Java$\texttrademark$}
provides \emph{LinkedHashMap} objects that combine fast random access
to key-value pairs with efficient key traversing. %
} (Listing: \ref{listing:bisimalg:marksplitcopyphase:marking}, Lines:
\ref{bisimalg:marksplitcopyphase:markerWorker:foreach-splits-keys-begin}
- \ref{bisimalg:marksplitcopyphase:markerWorker:foreach-splits-keys-end})
allows the routine for determining predecessors of $S$ (Listing:
\ref{listing:bisimalg:marksplitcopyphase:marking}, line: \ref{bisimalg:marksplitcopyphase:markerWorker:fetching_P})
and checks whether they preserve (Eq. \ref{eq:bcond}) (Listing: \ref{listing:bisimalg:marksplitcopyphase:marking},
line: \ref{bisimalg:marksplitcopyphase:markerWorker:split_need_ident}).
Every predecessor of $S$, which violates (Eq. \ref{eq:bcond}), i.e.
it contains two or more non-bisimilar states, is added to $\mathcal{M}$,
and the states from $\overline{P}(a,S)\subsetneq P$ become marked
(Listing: \ref{listing:bisimalg:marksplitcopyphase:marking}, Lines:
\ref{bisimalg:marksplitcopyphase:markerWorker:splits-map-iter} -
\ref{bisimalg:marksplitcopyphase:markerWorker:states-mark}).

\subsubsection{Splitting\label{sub:Splitting}}

When all the blocks from $\mathcal{S}$ are examined, then the set
$\mathcal{M}$ contains all the blocks identified as requiring division.
Moreover, with every block $M\in\mathcal{M}$ there is a set of marked
states $ms(M)$ assigned. To split the block $M$ only the states
from $ms(M)$ need to be processed. In the simplest case, $M$ is
divided into two parts $M\backslash ms(M)$ and $ms(M)$. In fact,
the set $ms(M)$ can be populated due to the many different splitters,
hence $ms(M)$ can be further subdivided into smaller subsets. Therefore,
in practice, $M$ is split into $M\backslash ms(M)$ and some number
of sub-blocks formed from $ms(M)$. 

In order to split the block $M\in\mathcal{M}$ first for every $v\in ms(M)$,
the state marker
\begin{equation}
m(v)\overset{df}{=}\{(l,\{id(P_{1}),\ldots,id(P_{k})\})|l=L(v,u)\wedge(v,u)\in\mathsf{T}\wedge u\in P_{i}\}\label{eq:state_mareker_def}
\end{equation}
is computed (Listing: \ref{listing:bisimalg:marksplitcopyphase:Splitting},
Lines: \ref{bisimalg:marksplitcopyphase:splitterWorker:markerComp-start}
- \ref{bisimalg:marksplitcopyphase:splitterWorker:markerComp-end}).
Then, states are grouped according to their markers in the \emph{subBlocksMap.
}

It is easy to observe%
\footnote{Note that the lack of stability would imply the existence of $P_{new}$
and $R$ so that $\overline{P}_{new}(l,R)\neq P_{new}$. Then there
would have to be $u\in P_{new}$ so that they would lead through the
transition $(u,t)$ labelled by $l$ to another block $P'$ different
than $P_{new}$, and would be at least one $v\in P_{new}$ that would
lead through the transition labelled by $l$ to $P_{new}$. However,
then $m(u)\neq m(v)$ therefore either $u$ or $v$ is not in $P_{new}$.
Contradiction.%
} that every new sub-block $P_{new}$ stored as the value in \emph{subBlocksMap}
is stable with respect to every splitter $S$ which was in $\mathcal{S}$.
Therefore, the initial $M$ will be replaced in $\mathcal{P}$ by
all the newly formed blocks stored as the values in \emph{subBlockMap}
(see Listing: \ref{listing:bisimalg:marksplitcopyphase:Splitting},
Lines: \ref{bisimalg:marksplitcopyphase:splitterWorker:allButTheLargestOneCond1}-\ref{bisimalg:marksplitcopyphase:splitterWorker:allButTheLargestOneCond1_addM},
\ref{bisimalg:marksplitcopyphase:splitterWorker:allButTheLargestOneCond2}-\ref{bisimalg:marksplitcopyphase:splitterWorker:aux-str-update})
and the block $M\backslash ms(M)$ (see: Listing: \ref{listing:bisimalg:marksplitcopyphase:Splitting},
line: \ref{bisimalg:marksplitcopyphase:splitterWorker:stateRemoval}).
Through this update operation, after the completion of the processing
elements stored in $\mathcal{M}$, the partition $\mathcal{P}$ becomes
stable with respect to its previous version, as it was used at the
beginning of the \emph{Mark-Split-Copy} loop. Of course, there is
no guarantee that it is stable with respect to itself. Therefore the
\emph{Mark-Split-Copy} loop is repeated until $\mathcal{M}\neq\varnothing$. 

If, after the reduction (Listing: \ref{listing:bisimalg:marksplitcopyphase:Splitting},
Line: \ref{bisimalg:marksplitcopyphase:splitterWorker:stateRemoval}),
$M$ becomes empty, it is removed from $\mathcal{P}$ (Listing: \ref{listing:bisimalg:marksplitcopyphase:Splitting},
line: \ref{bisimalg:marksplitcopyphase:splitterWorker:emptyM-removal}),
otherwise it remains in $\mathcal{P}$. If $M$ is not empty, then
the associated set $ms(M)$ of marked states is emptied (Listing:
\ref{listing:bisimalg:marksplitcopyphase:Splitting}, line: \ref{bisimalg:marksplitcopyphase:splitterWorker:clearingMarkedStatesSet}).
All the auxiliary structures are updated at the end of the procedure
(Listing: \ref{listing:bisimalg:marksplitcopyphase:Splitting}, line:
\ref{bisimalg:marksplitcopyphase:splitterWorker:aux-str-update}). 

\IncMargin{1em}  
\LinesNumbered  
\begin{algorithm}[h]

\texttt{Splitting($M$)}

\texttt{~~}\emph{subBlocksMap}\texttt{ $\leftarrow\varnothing$
\label{bisimalg:marksplitcopyphase:splitting:init-1} }

\texttt{~~$ $}\emph{markersMap}\texttt{ $\leftarrow\varnothing$
\label{bisimalg:marksplitcopyphase:splitting:init-2}}

\texttt{~~parallel for $v\in ms(M)$\label{bisimalg:marksplitcopyphase:splitterWorker:markerComp-start} }

\texttt{~~~~}\emph{markersMap.put}\texttt{$(v,\{(l,\{i_{1},\ldots,i_{k}\})|l=L(v,u)\wedge$}

\texttt{~~~~~~$\wedge i_{j}=$ }\emph{stateToBlockId}($\id(u)$)\texttt{$\})$\label{bisimalg:marksplitcopyphase:splitterWorker:markerComp-end}}

\texttt{~~parallel for $v\in ms(M)$ \label{bisimalg:marksplitcopyphase:splitterWorker:groupingByMark-start}}

\texttt{~~~~$m\leftarrow$}\emph{ markersMap.get($v$)}

\texttt{~~~~$ $}\emph{subBlocksMap.updateValueSet($m,v$)}\texttt{\label{bisimalg:marksplitcopyphase:splitterWorker:groupingByMark-end}}

\texttt{~~$M\leftarrow M\backslash\ms(M)$\label{bisimalg:marksplitcopyphase:splitterWorker:stateRemoval}}

\texttt{~~if }\emph{this is the second and subsequent pass of MarkSplitCopy
loop}\texttt{ \label{bisimalg:marksplitcopyphase:splitterWorker:theLPFinding}}

\texttt{~~~~$\smax$ $\leftarrow$ $\max\{$}\emph{subBlockMap.values}\texttt{$,M\}$\label{bisimalg:marksplitcopyphase:splitterWorker:choosingTheLargestPart}}

\texttt{~~else \label{bisimalg:marksplitcopyphase:splitterWorker:else}}

\texttt{~~~~$\smax$ $\leftarrow$ null}

\texttt{~~if ($M\neq\varnothing)$}

\texttt{~~~~$\ms(M)\leftarrow\varnothing$\label{bisimalg:marksplitcopyphase:splitterWorker:clearingMarkedStatesSet}}

\texttt{~~~~if ($M\neq\smax$) \label{bisimalg:marksplitcopyphase:splitterWorker:allButTheLargestOneCond1}}

\texttt{~~~~~~$ $add($\mathcal{S},M$)\label{bisimalg:marksplitcopyphase:splitterWorker:allButTheLargestOneCond1_addM}}

\texttt{~~else}

\texttt{~~~~remove($\mathcal{P},M$)\label{bisimalg:marksplitcopyphase:splitterWorker:emptyM-removal}}

\texttt{~~parallel for $B\in$ }\emph{subBlocksMap.values}\texttt{
\label{bisimalg:marksplitcopyphase:splitterWorker:allButTheLargestOneCond2}}

\texttt{~~~~if $\, B\neq\smax$ \label{bisimalg:marksplitcopyphase:splitterWorker:SplitterUpdate:if}}

\texttt{~~~~~~add($\mathcal{S},B$)\label{bisimalg:marksplitcopyphase:splitterWorker:SplitterUpdate}}

\texttt{~~~~add($\mathcal{P},B$)\label{bisimalg:marksplitcopyphase:splitterWorker:PartitionUpdate}}

\texttt{~~~~update }\emph{blockById}\texttt{, }\emph{nextStateToBlockId}\texttt{\label{bisimalg:marksplitcopyphase:splitterWorker:aux-str-update}}

\NoCaptionOfAlgo
\medskip \caption{\textbf{Listing \thealgocf: } Bisimulation algorithm - splitter routine}

\label{listing:bisimalg:marksplitcopyphase:Splitting}

\end{algorithm}
\DecMargin{1em}

\subsubsection{Process all but the largest one \label{sub:Process-all-but-the-largest-one}}

During the first turn of the \emph{Mark-Split-Copy} loop $\mathcal{P}=\mathcal{S}$,
i.e. all the blocks formed in the first initial phase are used as
splitters. In other words, the \emph{Marking()} procedure needs to
test the stability of every block $P\in\mathcal{P}$ with respect
to any other block in $\mathcal{P}$. During the second execution
and the subsequent ones, stability needs to be examined only with
respect to the changed blocks. Moreover, not all altered blocks need
to be added to the splitter set. That is because of one important
improvement proposed by \citep{PaigeTarjan1987tpra}, which follows
\emph{Hopcroft's} algorithmic strategy \emph{``process the smaller
half''} \citep{Aho1974tdaa,Hopcroft1971anln}. However, due to the
concurrent nature of this algorithm, \emph{Hopcroft's} principle cannot
be applied directly and needs to be ``parallelized''. Thus, the
new version of this principle ``\emph{Process all but the largest
one'' }means that all the blocks are formed as \emph{subBlockMap}
values, and $M$ after state removal (see Listing: \ref{listing:bisimalg:marksplitcopyphase:Splitting},
Lines: \ref{bisimalg:marksplitcopyphase:splitterWorker:allButTheLargestOneCond1},
\ref{bisimalg:marksplitcopyphase:splitterWorker:allButTheLargestOneCond2})
except the largest one of them, are added to $\mathcal{S}$. This
optimization is possible because every block $M$, being a subject
of \emph{Splitting(),} was the splitter itself in the past or is the
result of splitting a block that was a splitter. 

To illustrate the use of the \emph{process all but the largest one}
strategy, let us consider the case where, after the $k$-th pass of
the \emph{Mark-Split-Copy} loop $(k>1)$, $\mathcal{P}$ contains
blocks $P_{1},P_{2}$ so that $\oute(P_{1})\cap R_{1}\neq\varnothing$
and $\oute(P_{2})\cap R_{1}\neq\varnothing$, and there exists $R_{0}:R_{1}\subseteq R_{0}$
so that $R_{0}$ was used previously as a splitter . Then $P_{1}$
and $P_{2}$ are stable with respect to $R_{1}$, i.e. it holds that
$\overline{P}_{1}(a,R_{1})=P_{1}$ and $\overline{P}_{2}(b,R_{1})=P_{2}$
(Fig. \ref{fig:mark-split-copy-example-1}). 

\begin{figure}
\begin{centering}
\includegraphics[scale=0.8]{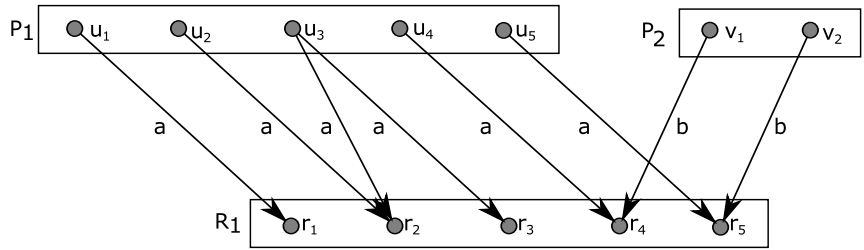}
\par\end{centering}

\caption{Mark-Split-Copy loop example - blocks $P_{1}$ and $P_{2}$ are stable
with respect to $R_{1}$}

\label{fig:mark-split-copy-example-1}
\end{figure}

Unfortunately, during the \emph{Split} step, the block $R_{1}$ is
divided into three different blocks $R_{11},R_{12}$ and $R_{13}$,
thus at the end of the $k$-th pass of the loop instead of $R_{1}$,
there are $R_{11},R_{12}$ and $R_{13}$ in $\mathcal{P}$. The stability
property is lost i.e. $\overline{P}_{1}(a,R_{11})\subsetneq P_{1}$,
$\overline{P}_{1}(a,R_{12})\subsetneq P_{1}$, $\overline{P}_{1}(a,R_{13})\subsetneq P_{1}$,
$\overline{P}_{2}(b,R_{12})\subsetneq P_{2}$ and $\overline{P}_{2}(b,R_{13})\subsetneq P_{2}$.
Therefore, it must be restored. Thus, according to the strategy of
\emph{process all but the largest one,} $R_{12}$ and $R_{13}$ were
added to the set of splitters $\mathcal{S}$ ($R_{11},R_{13}$ would
be equally good). Then, during the subsequent \emph{Mark} step ($k+1$-th
pass of the loop) both blocks $P_{1}$ and $P_{2}$ are added to $\mathcal{M}$,
and the vertices $u_{3},u_{4},u_{5},v_{1},v_{2}$ are marked (Listing:
\ref{listing:bisimalg:marksplitcopyphase:marking}, Lines: \ref{bisimalg:marksplitcopyphase:markerWorker:states-mark}
- \ref{bisimalg:marksplitcopyphase:markerWorker:foreach-splits-keys-end}).
Then, the blocks $P_{1}$ and $P_{2}$ are processed by the \emph{Splitting()}
procedure (Listing: \ref{listing:bisimalg:marksplitcopyphase:Splitting})
and the markers $\mv$ for the marked vertices are computed (Listing:
\ref{listing:bisimalg:marksplitcopyphase:Splitting}, Lines: \ref{bisimalg:marksplitcopyphase:splitterWorker:markerComp-start}
- \ref{bisimalg:marksplitcopyphase:splitterWorker:markerComp-end}).
There are $m(u_{3})=\{(a,\{R_{1},R_{2}\})\}$, $m(u_{4})=\{(a,\{R_{2}\})\}$,
$m(u_{5})=\{(a,\{R_{3}\})\}$, $m(v_{1})=\{(b,\{R_{2}\})\}$, $m(v_{2})=\{(b,\{R_{3}\})\}$
(Figure: \ref{fig:mark-split-copy-example-2}). 

\begin{figure}
\begin{centering}
\includegraphics[scale=0.8]{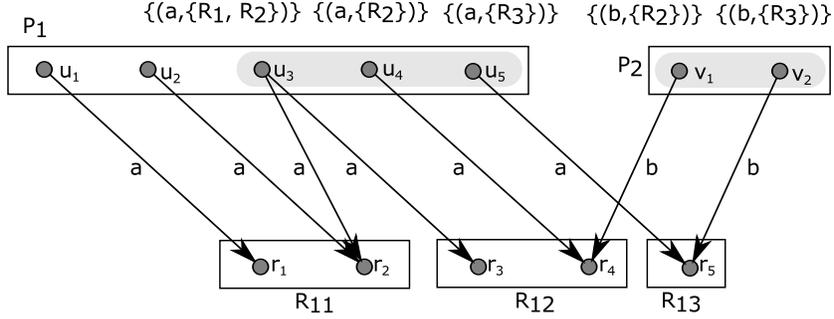}
\par\end{centering}

\caption{\emph{Mark-Split-Copy} loop example - blocks $P_{1}$ and $P_{2}$
are added to $\mathcal{M}$. Markers for the marked vertices (vertices
on the grayed background) are determined.}

\label{fig:mark-split-copy-example-2}
\end{figure}

Next, the blocks $P_{1}$ and $P_{2}$ are split, so that $P_{11}\leftarrow P_{1}\backslash ms(P_{1})$,
$P_{12}=\{u\in P_{1}|m(u)=\{(a,\{R_{1},R_{2}\})\}\}$, $P_{13}=\{u\in P_{1}|m(u)=\{(a,\{R_{2}\})\}\}$,
$P_{14}=\{u\in P_{1}|m(u)=\{(a,\{R_{3}\})\}\}$, $P_{21}=\{v\in P_{2}|m(v)=\{(b,\{R_{2}\})\}\}$
and finally $P_{22}=\{v\in P_{2}|m(v)=\{(b,\{R_{3}\})\}\}$ (Figure:
\ref{fig:mark-split-copy-example-3}). The blocks $P_{1}$ and $P_{2}$
are replaced in $\mathcal{P}$ by $P_{11},P_{12},P_{13},P_{14},P_{21},P_{22}$
(Listing: \ref{listing:bisimalg:marksplitcopyphase:Splitting}, Lines:
\ref{bisimalg:marksplitcopyphase:splitterWorker:emptyM-removal},
\ref{bisimalg:marksplitcopyphase:splitterWorker:PartitionUpdate}),
and once again $P_{1i},P_{2j}$ blocks are stable with respect to
the corresponding splitter blocks $R_{1i}$. Then, according to the
proposed strategy $P_{12},P_{13},P_{14}$ and $P_{22}$ are added
to $\mathcal{S}$ (Listing: \ref{listing:bisimalg:marksplitcopyphase:Splitting},
line: \ref{bisimalg:marksplitcopyphase:splitterWorker:SplitterUpdate}). 

\begin{figure}
\begin{centering}
\includegraphics[scale=0.8]{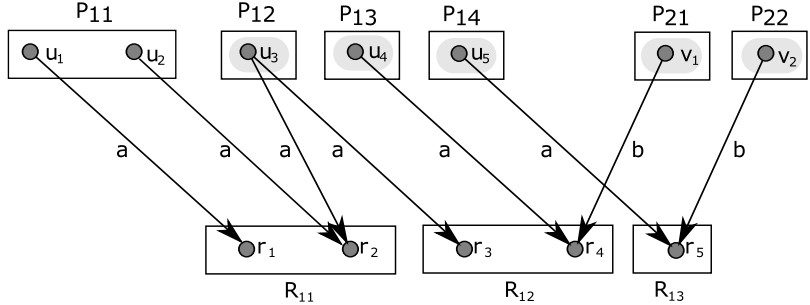}
\par\end{centering}

\caption{\emph{Mark-Split-Copy} loop example - blocks $P_{1}$ and $P_{2}$
are split}

\label{fig:mark-split-copy-example-3}
\end{figure}

Let us note that the markers for $u_{1}$ and $u_{2}$ (although they
are not computed) are $m(u_{1})=m(u_{2})=\{(a,\{R_{12}\})\}$. Thus,
$P_{11}$ could be defined as $P_{11}=\{u\in P|m(u)=\{(a,\{R_{12}\})\}$.
It is possible because the marker for $u_{1}$ and $u_{2}$ is different
from all other element markers in $P_{1}$. Thus, $\overline{P}_{11}(a,R_{11})=P_{11}$.
In other words, even if $R_{11}$ had been added to the splitters
set $\mathcal{S}$ the partitioning of $P_{1}$ would not change.
Regularity observed in the example can be explained more formally. 
\begin{theorem}
If the blocks $P_{1},\ldots,P_{k}$ are split with respect to the
splitter $R$, and it holds that $R$ was a part of a splitter block
in the past, the arbitrary selected sub-block of $R$ does not need
to be added to the splitter set $\mathcal{S}$. Of course, due to
performance reasons, the largest block is always omitted.\end{theorem}
\begin{proof}
Let $P_{1},\ldots,P_{k}$ be initially stable with respect to some
$R$, i.e. for every $a\in\mathsf{A}$ it holds that $\overline{P}_{i}(a,R_{1})=P_{i}$.
Let us assume that, as the result of the algorithm, $R$ is divided
into $R_{1}\cup\ldots\cup R_{r}$ and every $P_{i}$ is divided into
disjoint $P_{i1}\cup\ldots\cup P_{iq_{i}}$. It holds that for every
splitter, $R_{2},\ldots,R_{r}$ and every $P_{i1},\ldots,P_{iq_{i}}$
is either $\overline{P}_{ij}(a,R_{l})=P_{ij}$ or $\overline{P}_{ij}(a,R_{l})\cap P_{ij}=\varnothing$.
For the purpose of contradiction, let us assume that $R_{1}$, as
the largest block, was not considered as a splitter, hence $\overline{P}_{\hat{i}\hat{j}}(a,R_{1})\neq P_{\hat{i}\hat{j}}$
and $\overline{P}_{\hat{i}\hat{j}}(a,R_{1})\cap P_{\hat{i}\hat{j}}\neq\varnothing$.
Of course, $\overline{P}_{\hat{i}\hat{j}}(a,R_{1})\subseteq P_{\hat{i}\hat{j}}$
then, there must exist some state of $p$ such that $p\in P_{\hat{i}\hat{j}}$
and $p\notin\overline{P}_{\hat{i}\hat{j}}(a,R_{1})$. Since, initially
for every $a\in\mathsf{A}$ it holds that $\overline{P}_{i}(a,R_{1})=P_{i}$,
this means that there must be a transition $p\overset{a}{\rightarrow}q$
where $q\in P_{i}$. However, $p\notin\overline{P}_{\hat{i}\hat{j}}(a,R_{1})$
implies that $q\notin$ $R_{1}$. This means that there is $R_{l}$
and $l=2,\ldots,r$ such that $q\in R_{l}$. But every $R_{2},\ldots,R_{l}$
was processed as a splitter, hence the state $p$ has been assigned
to a new block $B$ according to their label $m(p)$. Since $\overline{P}_{\hat{i}\hat{j}}(a,R_{1})\cap P_{\hat{i}\hat{j}}\neq\varnothing$
let $\hat{p}$ such that $\hat{p}\in\overline{P}_{\hat{i}\hat{j}}(a,R_{1})\cap P_{\hat{i}\hat{j}}$.
Following the reasoning above, it is easy to see (for the same reasons
for which $p\in B$) that $\hat{p}\notin B$. Thus $B\neq P_{\hat{i}\hat{j}}$.
However, due to the nature of the algorithm, the blocks after splitting
do not overlap each other, thus $B\cap P_{\hat{i}\hat{j}}=\varnothing$.
In other words, the assumption $p\in B$ implies $p\notin P_{\hat{i}\hat{j}}$.
Contradiction. 
\end{proof}

\subsubsection{Copy }

The third part of the \emph{Mark-Split-Copy} loop is the \emph{Copy()}
routine (Listing: \ref{listing:bisimalg:marksplitcopyphase:copyWorker}).
It is responsible for swapping the \emph{nextStateToBlockId} and \emph{stateToBlockId}.
The copy step is required because the \emph{Splitting()} procedure
(Listing: \ref{listing:bisimalg:marksplitcopyphase:Splitting}) takes
the information about the current block assignments from \emph{stateToBlockId,}
whilst the new assignments used in the next iteration are stored in
\emph{nextStateToBlockId}. Such a solution reduces the locking overhead
when accessing the \emph{stateToBlockId} structure.

\section{Sequential complexity\label{sec:Sequential-complexity}}

One of the most important factors affecting the running time of the
algorithm are data structures. Very often, the same procedure can
run at different speeds for different data structures. The single-threaded
algorithm can run at full speed using more memory-efficient but sequential
data structures. The concurrent version of the same algorithm needs
highly parallel, but more memory consuming, data representation. Therefore,
the data structures used for the purpose of the analysis of the sequential
running time differ from those used for the purpose of the analysis
of the concurrent implementation. The current section deals with the
data structures suitable for the sequential processing, whilst the
concurrent implementation is discussed in (Sec. \ref{sub:Concurrent-algorithm}).

In the presented approach, the state and the block are the structures
that have their own unique \emph{id} automatically assigned to them
during creation. Thus, fetching \emph{id} for states and blocks takes
$O(1)$ of time. The mapping of \emph{ids} to blocks is provided by
the global map \emph{blockById} (Listing: \ref{listing:bisimalg:structures}).
Hence, fetching a block when its \emph{id} is known takes $O(1)$
on average \citep{Cormen2009ita}. The next two linked maps \emph{stateToBlockId}
and \emph{nextStateToBlockId} keep the membership relationship between
states (precisely their \emph{ids}) and blocks (the block \emph{ids}).
In other words, every state refers to the block to which it belongs.
Splitter blocks and blocks marked\emph{ }to be split are stored in
the queues $\mathcal{S}$ and $\mathcal{M}$ correspondingly (Listing
\ref{listing:bisimalg:structures}). Therefore, both operations \emph{add}
and \emph{poll} for $\mathcal{S}$ and $\mathcal{M}$ can be implemented
in the expected constant time $O(1)$. An \emph{initPartition} structure
can be implemented as a dynamic perfect hash map \citep{Dietzfelbinger1994dphu},
which maintains a double-linked list of its keys (\emph{linked hash
map}). It combines the ability to handle inserts and lookups in the
amortized time $O(1)$ together with fast iteration through the structure.
The keys stored within the \emph{initPartition} should be implemented
as a perfect hash map. Thus the key comparison can be performed in
the expected running time proportional to the length of the shorter
one. 

The set $\mathcal{P}$ represents (Listing \ref{listing:bisimalg:structures})
the current partition. In fact, maintaining $\mathcal{P}$ is superfluous.
That is because it can be easily obtained at the end of the algorithm
in $O(|\mathsf{S}|)$ by traversing \emph{stateToBlockId}, and no
decision during the course of the algorithm depends on $\mathcal{P}$.
Thus, placing it in the pseudocode is purely illustrative and their
updates wont be considered in the context of the overall running time
estimation. The auxiliary maps \emph{splitsMap} (Listing: \ref{listing:bisimalg:marksplitcopyphase:marking},
line: \ref{bisimalg:marksplitcopyphase:markerWorker:splitMapDef})
and \emph{initPartition} are implemented as the\emph{ }linked hash
map. The first one provides an \emph{updateValueSet(pm,~}$u$\emph{)}
method (Listing: \ref{listing:bisimalg:marksplitcopyphase:marking},
line: \ref{bisimalg:marksplitcopyphase:markerWorker:foreachS:stop}),
which adds the state $u$ to the set referred by \emph{pm}, or if
there is no such set, it first creates it, then adds $u$ into it.
Thus, the method can also be implemented in the time $O(1)$. The
sets held by \emph{splitsMap} as its values can also be implemented
as a linked hash map. The fact that some state $u\in B$ is marked
(Listing: \ref{listing:bisimalg:marksplitcopyphase:marking}, line:
\ref{bisimalg:marksplitcopyphase:markerWorker:states-mark}) can be
represented by adding an appropriate state id to the appropriate linked
hash set associated with that block. 

The \emph{Splitting()} routine starts with defining two auxiliary
linked hash maps: \emph{subBlocksMap }and \emph{markersMap}\texttt{$ $}.
As before, implementing them as linked hash maps allows estimation
of the expected running time of its constant operations. Like \emph{splitsMap,}
\emph{subBlocksMap} also provides the method \emph{updateValueSet()}
with the average running time $O(1)$. Items stored in \emph{markersMap}
are in the form of a pair \emph{(label, set of integers)}. The second
component of this pair can also be represented as a linked hash set.

\subsection{Initialization phase}

The initialization phase (Listing: \ref{listing:bisimalg:initialization-threephases})
is composed of three stages: the vertex signature initialization -
\emph{StateSignatureInit()}, the initial partition preparation \emph{-
PartitionInit()}, and the auxiliary structures initialization - \emph{AuxStructInit()}.
The first sub-procedure \emph{StateSignatureInit()} (Listing: \ref{listing:bisimalg:initialization-threephases},
Lines: \ref{code:initialization:VertexSigInit-begin} - \ref{code:initialization:VertexSigInit-end})
iterates through all the \emph{LTS} graph's vertices, and builds signatures
for all of them. The signature creation requires visiting all the
outgoing edges of the vertex. Since, during the course of the whole
initialization phase, each edge of \emph{LTS} is visited once, the
total running time of \emph{StateSignatureInit} is $O(|\mathsf{T}|)$.
The lookup operations (Listing: \ref{listing:bisimalg:initialization-threephases}
, line: \ref{code:initialization:VertexSigInit-contains-sig}) on
\emph{PartitionInit }implemented as the perfect hash map have $O(k(v))$
running time, where $k(v)$ is the length of the $v$ signature%
\footnote{It is assumed that the $\sig$ structure (used as a key in \emph{PartitionInit}
map) has its own hash code incrementally built when $\sig$ is created.
Thus, the expected running time of two $\sig$'s comparison when they
are different is $O(1)$. When two $\sig$'s are the same, the comparison
takes a time proportional to the size of $\sig$.%
} i.e. $k(v)=\left|\sig(v)\right|$. Since every vertex is accessed
once, and the total length of keys are $\underset{v\in\mathsf{S}}{\sum}|\sig(v)|=|\mathsf{T}|$
thus the total expected running time introduced by the keys comparison
(Listing: \ref{listing:bisimalg:initialization-threephases} Lines:
\ref{code:initialization:VertexSigInit-contains-sig} - \ref{code:initialization:VertexSigInit-end})
is at most $O(|\mathsf{T}|)$. The second initialization sub-procedure
\emph{PartitionInit()} (Listing: \ref{listing:bisimalg:initialization-threephases},
Lines: \ref{code:initialization:PartitionInit-begin} - \ref{code:initialization:PartitionInit-end})
iterates through the vertices of the \emph{LTS} graph and performs
one $O(|\sig(v)|)$ hash map lookup (Listing: \ref{listing:bisimalg:initialization-threephases},
line: \ref{code:initialization:PartitionInit-get-sig}) and two constant
time operations (Listing: \ref{listing:bisimalg:initialization-threephases},
Lines: \ref{code:initialization:PartitionInit-block-update} - \ref{code:initialization:VertexSigInit-end}).
Since the total length of the compared state signatures is $|\mathsf{T}|$,
the overall expected running time of \emph{PartitionInit()} is $O(|T|)$. 

Finally, the third sub-procedure \emph{AuxStructInit()} iterates through
the initial partitions (Listing: \ref{listing:bisimalg:initphase},
Lines: \ref{code:initializationPhase:AuxStructInit-iter-begin} -
\ref{code:initializationPhase:AuxStructInit-iter-end}) and adds them
(Listing: \ref{listing:bisimalg:initialization-threephases}, Lines:
\ref{code:initialization:AuxStructInit-begin} - \ref{code:initialization:AuxStructInit-end})
to other hash maps. Since the maximal number of partitions is limited
by $|\mathsf{S}|$ the time complexity of \emph{AuxStructInit()} is
$O(|\mathsf{S}|)$. Therefore, the overall expected sequential running
time of \emph{InitializationPhase()} (Listing: \ref{listing:bisimalg:initphase})
is $O(|\mathsf{T}|)+O(|\mathsf{T}|)+O(|\mathsf{S}|)=O(|\mathsf{T}|)$.

\subsection{MarkSplitCopy phase}

The second phase of the algorithm \emph{MarkSplitCopyPhase()} (Listing:
\ref{listing:bisimalg:marksplitcopyphase}) consists of three repetitive
subsequent steps, marking, splitting and copying. The first two steps
consist of sequences of calls \emph{Marking()} and \emph{Splitting()}.
For clarity of consideration, first the running times of \emph{Marking(),}
Splitting\emph{()} and \emph{Copy()} sub-routines are discussed. Then,
the overall running time of \emph{MarkSplitCopyPhase}() is (Listing:
\ref{listing:bisimalg:marksplitcopyphase}) examined.

\subsubsection{Marking}

The computational complexity of the \emph{Marking()} procedure (Listing:
\ref{listing:bisimalg:marksplitcopyphase:marking}) is determined
by the two loops \emph{for}. The first one (Listing: \ref{listing:bisimalg:marksplitcopyphase:marking},
Lines: \ref{bisimalg:marksplitcopyphase:markerWorker:foreachv-in-S}
- \ref{bisimalg:marksplitcopyphase:markerWorker:foreachS:stop}) is
responsible for \emph{splitsMap} preparation, whilst the second one
(Listing: \ref{listing:bisimalg:marksplitcopyphase:marking}, Lines:
\ref{bisimalg:marksplitcopyphase:markerWorker:foreach-splits-keys-begin}
- \ref{bisimalg:marksplitcopyphase:markerWorker:foreach-splits-keys-end})
shall select the blocks that are eligible for $\mathcal{M}$. Inside
the first loop, there is another \emph{for} loop (Listing: \ref{listing:bisimalg:marksplitcopyphase:marking},
Lines: \ref{bisimalg:marksplitcopyphase:markerWorker:foreachS:begin}
- \ref{bisimalg:marksplitcopyphase:markerWorker:foreachS:stop}) going
through all the states preceding the given state $s\in S$ in the
sense of $\mathsf{T}$. Thus, the line in which the partition marking
is assigned (Listing: \ref{listing:bisimalg:marksplitcopyphase:marking},
Line: \ref{bisimalg:marksplitcopyphase:markerWorker:foreachS:partMarker})
is executed as many times as there are incoming edges into $S$. In
particular, for a single $v\in S$ the loop is executed $O(|\ine(v)|)$
times. Hence, the running time of the first loop is $O(\underset{v\in S}{\sum}|\ine(v)|)$. 

The \emph{splitsMap} structure maps keys in the form of a pair \emph{(label,~partitionid)}
to the values, which are represented by the blocks of states. During
every single pass of the first \emph{for} loop a single state $u$
to some state block, stored\emph{ }as the value in\emph{ splitsMap,}
is added (Listing: \ref{listing:bisimalg:marksplitcopyphase:marking},
line: \ref{bisimalg:marksplitcopyphase:markerWorker:foreachS:stop}).
Thus, the total number of elements in all the blocks stored as the
values in \emph{splitsMap} does not exceed $O(\underset{v\in S}{\sum}|\ine(v)|)$.
For this reason, within the second \emph{for} loop (Listing: \ref{listing:bisimalg:marksplitcopyphase:marking},
Lines: \ref{bisimalg:marksplitcopyphase:markerWorker:foreach-splits-keys-begin}
- \ref{bisimalg:marksplitcopyphase:markerWorker:foreach-splits-keys-end})
the \emph{mark()} routine is called at most $O(\underset{v\in S}{\sum}|\ine(v)|)=O(|\ine(S)|)$
times. Since the running time of \emph{mark()} is $O(1)$, and the
other operations within the second \emph{for} loop can also be implemented
in $O(1)$, then the overall running time of the second \emph{for}
loop, and hence, the result of the entire \emph{Marking()} procedure
(Listing: \ref{listing:bisimalg:marksplitcopyphase:marking}), is
$O(\underset{v\in S}{\sum}|\ine(v)|)$.

\subsubsection{Splitting}

The first two instructions of the \emph{Splitting()} procedure (Listing:
\ref{listing:bisimalg:marksplitcopyphase:Splitting}, Lines: \ref{bisimalg:marksplitcopyphase:splitting:init-1}
- \ref{bisimalg:marksplitcopyphase:splitting:init-2}) are responsible
for the initialization of the auxiliary structures and can be implemented
in $O(1)$. The running time of the loop (Listing: \ref{listing:bisimalg:marksplitcopyphase:Splitting},
Lines: \ref{bisimalg:marksplitcopyphase:splitterWorker:markerComp-start}
- \ref{bisimalg:marksplitcopyphase:splitterWorker:markerComp-end})
depends on the size of $\ms(M)\subseteq M$ and the time required
to create state markers (Listing: \ref{listing:bisimalg:marksplitcopyphase:Splitting},
line: \ref{bisimalg:marksplitcopyphase:splitterWorker:markerComp-end}).
The time needed to create a single marker for $v$ depends on the
size of $\oute(v)$, thus the running time of this loop (Listing:
\ref{listing:bisimalg:marksplitcopyphase:Splitting}, Lines: \ref{bisimalg:marksplitcopyphase:splitterWorker:markerComp-start}
- \ref{bisimalg:marksplitcopyphase:splitterWorker:markerComp-end})
is expected to be $O(\underset{v\in ms(M)}{\sum}|\oute(v)|)\leq O(\underset{v\in M}{\sum}|\oute(v)|)$. 

The next loop \emph{for} (Listing: \ref{listing:bisimalg:marksplitcopyphase:Splitting},
Lines: \ref{bisimalg:marksplitcopyphase:splitterWorker:groupingByMark-start}
- \ref{bisimalg:marksplitcopyphase:splitterWorker:groupingByMark-end})
depending on $\ms(M)$ has the running time%
\footnote{The only exception is the \emph{updateValueSet()} method call, which
is discussed later on.%
} $O(|\ms(M)|)\leq O(|M|)$. Similarly, removing the marked vertices
from $M$ (Listing: \ref{listing:bisimalg:marksplitcopyphase:Splitting},
line: \ref{bisimalg:marksplitcopyphase:splitterWorker:stateRemoval})
does not need more operations than $|\ms(M)|$. The next few lines
(Listing: \ref{listing:bisimalg:marksplitcopyphase:Splitting}, Lines:
\ref{bisimalg:marksplitcopyphase:splitterWorker:theLPFinding} - \ref{bisimalg:marksplitcopyphase:splitterWorker:emptyM-removal})
can be implemented in the constant time $O(1)$. The only exception
is the choice of the largest block out of \emph{subBlockMap.values}
and $M$ (Listing: \ref{listing:bisimalg:marksplitcopyphase:Splitting},
line: \ref{bisimalg:marksplitcopyphase:splitterWorker:choosingTheLargestPart}),
which can be computed in time proportional to $|$\emph{subBlockMap.values}$|$$\leq|\ms(M)|\leq|M|$.
The last \emph{for} loop (Listing: \ref{listing:bisimalg:marksplitcopyphase:Splitting},
line: \ref{bisimalg:marksplitcopyphase:splitterWorker:allButTheLargestOneCond2}
- \ref{bisimalg:marksplitcopyphase:splitterWorker:aux-str-update})
also depends on the size of \emph{subBlockMap.values}. Thus, the operation
of adding blocks to $\mathcal{S}$ and $\mathcal{P}$ is executed
at most \emph{$|$subBlockMap.values}$|$ times. The only exception
is \emph{nextStateToBlockId} update, which entails visiting every
state $b\in B$. Since $|B_{1}\cup\ldots\cup B_{|\sbmv|}|$ $\leq|M|$,
then the loop visits at most $|M|$ states, and hence, its running
time is at most $O(|M|)$. Thus, finally, the overall running time
of the \emph{Splitting()} procedure is determined as:
\begin{equation}
O(\max\{|\underset{v\in M}{\sum}|\oute(v)|,|M|\})\label{eq:seq-splitting-running-time-estim}
\end{equation}

The \emph{Marking()} and \emph{Splitting()} procedures are called
within the loops iterating through the elements of $\mathcal{S}$
and $\mathcal{M}$ correspondingly (Listing: \ref{listing:bisimalg:marksplitcopyphase}).
For the purpose of the current consideration, all the three parallel
\emph{for} loops (including \emph{Copying()}) are treated as sequential
loops (Listing: \ref{listing:bisimalg:marksplitcopyphase}, Lines:
\ref{code:MarkSplitCopyPhase:mainLoop:confor-call-1}, \ref{code:MarkSplitCopyPhase:mainLoop:confor-call-2},
\ref{code:MarkSplitCopyPhase:mainLoop:confor-call-3}). 

The initialization phase prepares the initial partition refinement
$\mathcal{P}$ and sets all the blocks from $\mathcal{P}$ as splitters
i.e. $\mathcal{S}=\mathcal{P}$. Thus, during the first pass of the
loop (Listing: \ref{listing:bisimalg:marksplitcopyphase}, Lines:
\ref{code:MarkSplitCopyPhase:mainLoop:Begin} - \ref{code:MarkSplitCopyPhase:mainLoop:MarkingCall})
all the blocks $S_{i}\in\mathcal{P}$ are processed. Since every single
call of the \emph{Marking()} procedure with $S_{i}\in\mathcal{S}$
on its input needs $O(|\ine(S_{i})|)$ time, and the cardinality of
$\underset{S_{i}\in\mathcal{S}}{\bigcup}\ine(S_{i})$ is smaller than
$\left|\mathsf{T}\right|$, then the sequential running time of the
first \emph{for} loop (Listing: \ref{listing:bisimalg:marksplitcopyphase},
Lines: \ref{code:MarkSplitCopyPhase:mainLoop:Begin} - \ref{code:MarkSplitCopyPhase:mainLoop:MarkingCall})
is upper-bounded by: $ $
\begin{equation}
O(\underset{S_{i}\in\mathcal{S}}{\sum}\underset{v\in S_{i}}{\sum}|\ine(v)|)=O(|\mathsf{T}|)\label{eq:seq-total-marker-worker-est}
\end{equation}

The \emph{Marking()} calls are responsible for filling $\mathcal{M}$
- the set of marked blocks for splitting. Since each splitter can
mean that any block $M\in\mathcal{P}$ gets marked, the total number
of states in all the marked blocks is upper-bounded by $\underset{M\in\mathcal{M}}{\sum}\underset{u\in M}{\sum}|u|\leq|\mathsf{S}|\leq|\mathsf{T}|$.
Of course, \emph{Splitting()} is called for every marked block $M\in\mathcal{M}$.
Taking into account that the running time of every single \emph{Splitting()}
routine is $O(\max\{|\oute(M)|,|M|\})$, the sequential running time
of the first execution of the second \emph{for} loop (Listing: \ref{listing:bisimalg:marksplitcopyphase},
Lines: \ref{code:MarkSplitCopyPhase:mainLoop:confor-call-2} - \ref{code:MarkSplitCopyPhase:mainLoop:confor-splittingCall})
is: 
\begin{equation}
O(\max\{\underset{M\in\mathcal{M}}{\sum}\underset{u\in M}{\sum}|\oute(u)|,\underset{M\in\mathcal{M}}{\sum}|M|)\leq O(\max\{|\mathsf{T}|,|\mathsf{S}|\})\label{eq:lab}
\end{equation}

In the second and subsequent passes of the \emph{MarkSplitCopy} loop
(Listing: \ref{listing:bisimalg:marksplitcopyphase}, Lines: \ref{code:MarkSplitCopyPhase:mainLoop:Begin}
- \ref{code:MarkSplitCopyPhase:mainLoop:End}) the number of splitters
in $\mathcal{S}$ is less than the number of blocks in $\mathcal{P}$.
Since, during the second iteration, all the blocks from $\mathcal{P}$
are used as splitters at least once, during the splitting step the
rule \emph{process all but the largest one} can be applied. In order
to illustrate the impact of this rule on the running time of the algorithm,
let us consider $S^{2}$ added to $\mathcal{S}$ as the result of
the first pass of the \emph{MarkSplitCopy} loop (Figure: \ref{fig:MarkSplitCopyLoop:execution:tree}).
Let us suppose that $M_{1}^{2},\ldots,M_{i}^{2}$, (where $|\mathcal{P}|\geq i\geq2$)
are the blocks marked as a result of processing $S^{2}$. 

\begin{figure}
\begin{centering}
\includegraphics[scale=1.1]{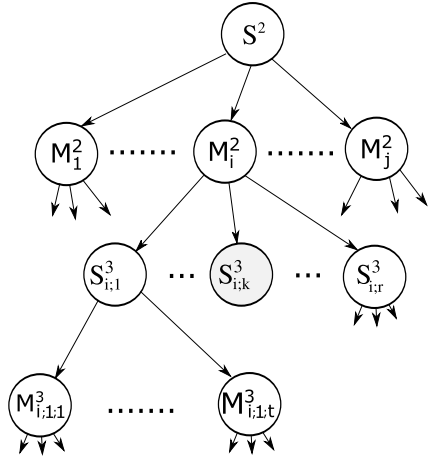}
\par\end{centering}

\caption{MarkSplitCopy loop execution tree}
\label{fig:MarkSplitCopyLoop:execution:tree}
\end{figure}

Then, every marked block $M_{j}^{2}$ is divided into $|M_{j}^{2}|\geq r\geq2$
sub-blocks: $S_{i;1}^{3},\ldots,S_{i;r}^{3}$. Let us assume that
$S_{i;k}^{3}$ is the largest of them. Then it holds that $|S_{i;k}^{3}|\geq\frac{|M_{j}^{2}|}{r}$,
and each block $S_{i;1}^{3},\ldots,S_{i;k-1}^{3},S_{i;k+1}^{3},\ldots,S_{i;r}^{3}$
is smaller than $\frac{|M_{j}^{2}|}{2}$. The block $S_{i;k}^{3}$
as the largest one is not further processed. Hence, even in the worst
case, the blocks intended for further processing are smaller (or equal)
than half of the block from which they were separated%
\footnote{In fact, the division into several blocks means that very often the
size of each of them is substantially less than half of the divided
block.%
}. This leads to the conclusion that the single $v\in\mathsf{S}$ can
be assigned during the course of the algorithm to at most $\log|\mathsf{S}|$
different blocks. In other words, every $v\in\mathsf{S}$ is at most
$\log|\mathsf{S}|$ times processed by \emph{Splitting()}. Therefore,
the total time of the \emph{Splitting()} procedure in the second and
subsequent iterations of the \emph{MarkSplitCopy} loop is 
\begin{equation}
\underset{v\in\mathsf{S}}{\sum}\max\{|\oute(v)|,1\}\cdot\log|\mathsf{S}|=\max\{|\mathsf{T}|,|\mathsf{S}|\}\cdot\log|\mathsf{S}|\label{eq:lab2}
\end{equation}

Therefore, the total time of the presented algorithm with respect
to the \emph{Splitting()} procedure for the first, second and subsequent
iterations of the \emph{MarkSplitCopy} loop is: 

\begin{equation}
O(\max\{|\mathsf{T}|,|\mathsf{S}|\}\cdot(1+\log|\mathsf{S}|))\label{eq:lab3}
\end{equation}

In the \emph{Splitting()} procedure, the linked hash map structure
is used not only as the data container, but also as part of a key
structure in the other hash map. This happens when \emph{updateValueSet(m,v)}
on \emph{subBlocksMap} is called (Listing: \ref{listing:bisimalg:marksplitcopyphase:Splitting},
Line: \ref{bisimalg:marksplitcopyphase:splitterWorker:groupingByMark-end}).
Thus, in fact, the operation \emph{updateValueSet($m,v$)} does not
run in the time $O(1)$ as would happen if the key comparison was
a simple operation, but depends on the size of the key in the form
of the $k+1$ tuple $m=(l,\{i_{1},\ldots,i_{k}\})$. Assuming that
the set of integers $\{i_{1},\ldots,i_{k}\}$ is implemented as a
\emph{linked hash map}, the comparison of two keys is expected to
take $O(k)$. Thus, for every $v\in\mathsf{S}$ the single call \emph{updateValueSet($m,v$)}
takes at most $O(k)$, where $k=|\oute(v)|$. Since $v\in\ms(M)$,
and during the course of the whole algorithm every $v\in\mathsf{S}$
is at most in $\log|\mathsf{S}|$ blocks, the total time of \emph{updateValueSet($m,v$)}
calls is: 
\begin{equation}
O(\underset{v\in\mathsf{S}}{\sum}|\oute(v)|\cdot\log|\mathsf{S}|)=O(|\mathsf{T}|\cdot\log|\mathsf{S}|)\label{eq:lab4}
\end{equation}
In other words, despite the fact that the key is a collection itself,
its use does not affect the overall asymptotic running time of the
algorithm. 

Splitters, i.e. the blocks processed by the \emph{Marking()} procedure
are added to $\mathcal{S}$ by the \emph{Splitting()} procedure as
a result of partitioning blocks from $\mathcal{M}$. Thus, the single
$v$ is a member of the splitter $S$ as frequently as it is processed
by the \emph{Splitting()} procedure. Therefore, $v$ as an element
of some splitter $S\in\mathcal{S}$ is under consideration of the
\emph{Marking()} procedure at most $\log|\mathsf{S}|$ times. For
this reason, the total time of \emph{Marking()} in the second and
subsequent iterations is: 
\begin{equation}
\underset{v\in\mathsf{S}}{\sum}|\ine(v)|\cdot\log|\mathsf{S}|=|\mathsf{T}|\cdot\log|\mathsf{S}|\label{eq:lab5}
\end{equation}
Thus, the overall running time of the presented algorithm with respect
to \emph{Marking()} is 
\begin{equation}
O(|\mathsf{T}|\cdot(1+\log|\mathsf{S}|))\label{eq:lab6}
\end{equation}

The least complicated is the \emph{Copying()} procedure. It is responsible
for adding pairs \emph{(state, id)} into \emph{nextStateToBlockId}
when the given state has changed its block assignment. Since every
state can be in at most $\log|\mathsf{S}|$ marked blocks, the \emph{Copying()}
procedure needs to copy at most $|\mathsf{S}|\cdot\log|\mathsf{S}|$
pairs during the course of the algorithm. Thus, the overall running
time of \emph{Copying()} is $O(|\mathsf{S}|\cdot\log|\mathsf{S}|)$.
Moreover, assuming that it is faster on the given platform to replace
\emph{stateToBlockId} by \emph{nextStateToBlockId}, \emph{Copying()}
can be reduced to a simple assignment with the constant running time~$O(1)$. 

In conclusion, the expected overall sequential running time of the
whole algorithm is: 
\begin{equation}
O(|\mathsf{T}|\cdot\log|\mathsf{S}|)\label{eq:lab7}
\end{equation}
which is the sum of initialization $O(|\mathsf{T}|)$, overall expected
running time of \emph{Splitting()} and \emph{Marking() }is $O(|\mathsf{T}|\cdot(1+\log|\mathsf{S}|))$,
whilst copying data from \emph{nextStateToBlockId} to \emph{stateToBlockId}
takes at most $O(|\mathsf{S}|\cdot\log|\mathsf{S}|)$. 

The amount of memory used by the sequential algorithm depends on the
amount of memory used by its data structures. Thus, assuming that
that memory occupied by a hash map depends linearly on the number
of elements stored in it \citep{Cormen2009ita}, the global data structures:
\emph{initPartition, blockById, stateToBlockId, nextStateToBlockId,
$\mathcal{M}$ }and\emph{ $\mathcal{S}$ }need at most \emph{$O(|\mathsf{S}|)$}
of memory. The total size of the keys used by \emph{splitsMap} (Listing:
\ref{listing:bisimalg:marksplitcopyphase:marking}, Line: \ref{bisimalg:marksplitcopyphase:markerWorker:foreachS:partMarker})
is at most $O(|\mathsf{T}|)$, and similarly the total size of the
keys used by \emph{subBlocksMap} (Listing: \ref{listing:bisimalg:marksplitcopyphase:Splitting},
Line: \ref{bisimalg:marksplitcopyphase:splitterWorker:groupingByMark-end})
is $O(|\mathsf{T}|)$. The size of \emph{splitsMap, markersMap }and\emph{
subBlocksMap }is at most $O(|\mathsf{S}|)$. Thus, the overall amount
of memory used by the algorithm in the given moment of time is $O(|\mathsf{S}|+|\mathsf{T}|)$.

\section{Concurrent algorithm\label{sub:Concurrent-algorithm}}

The presented algorithm has been designed to use concurrent objects
(concurrent data structures) \citep{Herlihy2008taom,Moir2007cds}.
Concurrent objects try to preserve the running time known from the
sequential data structures while maintaining a high level of concurrency.
In situations where the number of processors is limited, and the overall
sequential time of the method execution is much larger than the synchronization
time, such objects seem to be a good practical approximation of fully
parallel data structures in the abstract \emph{PRAM} (\emph{Parallel}
\emph{Random} \emph{Access} \emph{Machines}) model. The practical
implementation guidelines, together with the preliminary results of
the experimental implementation, can be found in (Sec. \ref{sub:Concurrent-implementation}). 

Considering the concurrent algorithm, the question arises as to what
extent the algorithm could be parallelized? In other words, to what
asymptotic running time may the proposed algorithm tend to? In order
to answer this question (Sec. \ref{sub:Concurrent-complexity}), let
us assume that all the data collections used by the algorithm are
fully parallel accessible arrays. As the parallel computation model,
the shared memory model (\emph{PRAM}) is adopted. Of course, such
an assumption means the fully parallel implementation would require
a very large amount of memory. The number of processors that would
be able to simultaneously handle these arrays also must be large.
Therefore, the main objective of discussing parallelization capabilities
(Sec. \ref{sub:Concurrent-complexity}) is to determine the lower
bound of the parallel running time of the presented construction.
Despite this, in some specific cases the fully parallel implementation
of the presented algorithm might be of interest for practitioners.

\subsection{Capabilities of parallelization \label{sub:Concurrent-complexity}}

For the purpose of studying to what extent the presented algorithm
could be parallelized, all the data structures need to be implemented
as arrays of references. Therefore, it is assumed that maps such as
\emph{initPartition}, \emph{blockById}, \emph{stateToBlockId}, and
\emph{nextStateToBlockId}, and queues $\mathcal{M}$ and $\mathcal{S}$
are implemented as directly addressed tables \citep[p. 254]{Cormen2009ita}.
The methods \emph{Marking()} and \emph{Splitting()} use their own
local maps \emph{splitsMap}, \emph{subBlocksMap} and \emph{markersMap}.
All of them also need to be represented as tables. All of these tables,
except \emph{initPartition} and \emph{subBlocksMap,} are naturally
indexable by numbers resulting from the algorithm. In the case of
\emph{initPartition} and \emph{subBlocksMap,} indexes have the form
of n-tuples, thus the actual indexes need to be calculated in parallel.
The algorithm allowing for conversion of an unsorted n-tuple of numbers
to the single numeric index can be found in (App. \ref{sec:parallel-index-creation-lema}).
In addition to the data structures, the \emph{LTS} graph also needs
to be in the form of an array. Hence, it is assumed that every state
$s\in\mathsf{S}$ provides the arrays \emph{s.out} and \emph{s.in}
of outgoing and incoming edges. Every edge provides reference to its
beginning state, ending state, and the label. Every partition block
$B\in\mathcal{P}$ also provides two arrays \emph{B.vert}$ $ and
\emph{B.mvert} that store the states belonging to $B$, and the states
that are marked by the \emph{Marking()} procedure correspondingly.
Moreover, it is assumed that for each label $l\in\mathsf{A}$ and
state $s\in\mathsf{S}$ an appropriate integer number (its unique
index) could be computed in the constant time $O(1)$. The provided
estimation uses the designations $\alpha$ as the maximal length of
the state signature,i.e. $\alpha=\underset{s\in\mathsf{S}}{\max}|\sig(s)|$,
and $\beta$ as the state maximal output degree i.e. $\beta=\underset{s\in\mathsf{S}}{\max}$$|$\emph{s.out}$|$.

\subsubsection{Initialization phase\label{sub:Complexity:Initialization-phase}}

Assuming that the algorithm has at its disposal at least $|\mathsf{S}|$
processors, then the first two subroutine calls can be executed fully
independently (Listing: \ref{listing:bisimalg:initphase}, Lines:
\ref{code:initializationPhase:StateSignatureInit-start} - \ref{code:initializationPhase:PartInitEnd}).
The $|\mathsf{A}|^{\beta}$ number of processors allows every cell
of the \emph{initPartition} array to be assigned to separate processors.
Therefore, $O(1)$ running time of \emph{AuxStructInit()} is paid
by a demand for $|\mathsf{A}|^{\beta}$ processors (Listing: \ref{listing:bisimalg:initialization-threephases},
Lines: \ref{code:initialization:AuxStructInit-begin} - \ref{code:initialization:AuxStructInit-end}),
although the actual work is $O(|\mathsf{S}|)$, since at most $|\mathsf{S}|$
cells in \emph{initPartition} are not empty. 

The first \emph{StateSignatureInit()} needs to prepare signatures
for the given $v\in\mathsf{S}$ (Listing: \ref{listing:bisimalg:initialization-threephases},
Line: \ref{code:initialization:VertexSigInit-sig-creation}). In fact,
the signature $\sig(v)$ needs to be a single integer uniquely calculated
on the basis of the \emph{ids} of \emph{labels}. Since the labels
can be easily identified by integers, the desired index can be calculated
in $O(\log\beta)$ using the $\beta$ processor (according to App.
\ref{sec:parallel-index-creation-lema}). The \emph{StateSignatureInit()}
is also responsible for creating new empty blocks (Listing: \ref{listing:bisimalg:initialization-threephases},
Line: \ref{code:initialization:VertexSigInit-end}). These blocks
should have uniquely assigned \emph{ids}. In the sequential case,
this was not a problem due to the use of a simple shared counter.
In the parallel case, however, to avoid synchronization during the
block instantiation, it assumed that the ids of blocks are determined
on the basis of the processor's id. Unfortunately, this solution increases
the range of block indices from $|\mathsf{S}|\log|\mathsf{S}|$ to
$|\mathsf{S}|^{2}$, which results in an increase in resource demand.
Since the \emph{initPartition} is implemented as an array, then the
other parts of the \emph{StateSignatureInit()} procedure (Listing:
\ref{listing:bisimalg:initialization-threephases}, Lines: \ref{code:initialization:VertexSigInit-contains-sig}
- \ref{code:initialization:VertexSigInit-contains-sig}) can be implemented
in $O(1)$.  

The second initialization procedure \emph{PartitionInit()} contains
three simple array operations. These are getting values from the \emph{initPartition}
array (Listing: \ref{listing:bisimalg:initialization-threephases},
Line: \ref{code:initialization:PartitionInit-get-sig}), inserting
references to $v\in\mathsf{S}$ into \emph{block.vert} (Listing: \ref{listing:bisimalg:initialization-threephases},
Line: \ref{code:initialization:PartitionInit-block-update}) and inserting
\emph{block.id} into \emph{stateToBlockId}. Thus, the whole procedure
\emph{PartitionInit()} takes $O(1)$ running time in parallel. Similarly,
\emph{AuxStructInit()} contains three simple table operations (Listing:
\ref{listing:bisimalg:initialization-threephases}, Lines: \ref{code:initialization:BlockById}
- \ref{code:initialization:AuxStructInit-end}). Thus, its overall
running time is $O(1)$. 

The method \emph{StateSignatureInit()} is called by \emph{InitializationPhase()}
in parallel for every state $s\in\mathsf{S}$ (Listing: \ref{listing:bisimalg:initphase},
Lines: \ref{code:initializationPhase:StateSignatureInit-start} -
\ref{code:initializationPhase:StateSignatureInit-end}), which requires
$|\mathsf{S}|$ processors. Since every single call of \emph{StateSignatureInit()}
needs $\beta$ processors and takes $O(\log\beta)$ of time, the running
time of \emph{InitializationPhase() }with respect to its first sub-procedure
is $O(\log\beta)$ and requires $\beta\cdot|\mathsf{S}|$ processors.
The second \emph{PartitionInit()} method does not contain any additional
complex operations inside, therefore it does not contribute to an
increase of the overall running time estimation of the \emph{InitializationPhase()}
procedure. Finally, \emph{AuxStructInit()} is executed in parallel
as many times as entries in the \emph{initPartition} array. Thus,
to obtain $O(1)$ running time, it needs $|\mathsf{A}|^{\alpha}$
parallel processors. 

In summary, the \emph{InitializationPhase()} could achieve $O(\log\beta)$
running time, using $max\{\beta\cdot|\mathsf{S}|,|\mathsf{A}|^{\alpha}\}$
processors. The actual work performed by processors is at most $O(\beta\cdot|\mathsf{S}|)$.
Taking into account that the running time of the main part of the
algorithm achieves $O(|\mathsf{S}|\log\beta)$ running time, the above
estimates could be relaxed in practice without affecting the final
estimates.

\subsubsection{MarkSplitCopy phase}

The running time of the \emph{MarkSplitCopyPhase()} method (Listing:
\ref{listing:bisimalg:marksplitcopyphase}) depends on the number
of the loop during execution (Listing: \ref{listing:bisimalg:marksplitcopyphase},
Lines: \ref{code:MarkSplitCopyPhase:mainLoop:Begin} - \ref{code:MarkSplitCopyPhase:mainLoop:End}),
and the parallel running time of its subroutine calls. The maximal
number of block splits during the course of the algorithm is $|\mathsf{S}|$.
Thus, it might happen that only one block split in every turn of the
loop is performed. Hence, in the worst case scenario the loop executes
$|\mathsf{S}|$ times. As shown below, every iteration of the loop
takes at most $O(\log\beta)$ time, and the overall parallel running
time of the \emph{MarkSplitCopyPhase()} method,\emph{ }and thus the
algorithm, is $|\mathsf{S}|O(\log\beta)$.

\paragraph{Marking}

The first considered subroutine is \emph{Marking().} It is called
in parallel for every block $S\in\mathcal{S}$. Because $|\mathcal{S}|\leq|\mathsf{S}|$,
at most $|\mathsf{S}|$ processors are needed to call \emph{Marking()}
in parallel (Listing: \ref{listing:bisimalg:marksplitcopyphase},
Lines: \ref{code:MarkSplitCopyPhase:mainLoop:confor-call-1} - \ref{code:MarkSplitCopyPhase:mainLoop:MarkingCall}).
Two parallel \emph{for} loops: the iteration through $s\in S$ (Listing:
\ref{listing:bisimalg:marksplitcopyphase:marking}, Line: \ref{bisimalg:marksplitcopyphase:markerWorker:foreachv-in-S})
and the iteration through $(u,s)\in$\emph{~s.in} need at most $\max\{|\mathsf{S}|,|\mathsf{T}|\}$
processors to be executed in parallel. Two operations: key creation
(Listing: \ref{listing:bisimalg:marksplitcopyphase:marking}, Line:
\ref{bisimalg:marksplitcopyphase:markerWorker:foreachS:partMarker})
and inserting the state's reference into the array (Listing: \ref{listing:bisimalg:marksplitcopyphase:marking},
Line: \ref{bisimalg:marksplitcopyphase:markerWorker:foreachS:stop})
need $O(1)$ running time%
\footnote{It is assumed that \emph{splitsMap} holds references to the arrays
of size $|\mathsf{S}|$%
}. 

To visit all the arrays stored in \emph{splitsMap} in parallel $|\mathsf{A}||\mathsf{S}|^{2}$,
parallel processors are required (Listing: \ref{listing:bisimalg:marksplitcopyphase:marking},
Line: \ref{bisimalg:marksplitcopyphase:markerWorker:foreach-splits-keys-begin}).
The next line (Listing: \ref{listing:bisimalg:marksplitcopyphase:marking},
Line: \ref{bisimalg:marksplitcopyphase:markerWorker:fetching_P})
fetches the block reference from the \emph{blockById} array and takes
$O(1)$ of time. Also, it decides whether $|B|>1$ can be performed
in asymptotically constant time%
\footnote{This is made possible by the following procedure. At the very beginning,
every processor $p_{i}\in\{p_{1},\ldots,p_{|\mathsf{S}|}\}$ checks
whether \emph{B.vert}$[i]$ contains a reference to a state object,
and if so writes its index in the shared variable $x$. Next, the
winning reference is copied to auxiliary \emph{tmp} $\leftarrow$\emph{
B.vert}$[x]$, and the winning cell is cleared \emph{B.vert}$[x]\leftarrow$
\emph{nil}. Then, once again, every processor $p_{i}\in\{p_{1},\ldots,p_{|\mathsf{S}|}\}$
is asked to write its index in the shared variable $y$ in case \emph{B.vert}$[i]\neq nil$.
Thus, after the restoration of \emph{B.vert}$[x]\leftarrow$ \emph{tmp,}
it holds that $|B|>1$ if and only if $x$ and $y$ are different
from $0$. %
} $O(1)$. 

Although, in the sequential case, the running time needed for deciding
whether $|B|>$ $|$\emph{splitsMap.get(pm)}$|$ was negligible (the
elements are inserted into \emph{splitsMap.get(pm)} and $B$ sequentially,
thus they can also be sequentially counted during the insertion),
in the parallel case, it needs to be taken into account. Since all
the elements are inserted into $B$ and \emph{splitsMap.get(pm)} in
parallel, there is no place where the shared counter could help. Therefore,
the condition $|B|>$ $|$\emph{splitsMap.get(pm)}$|$ needs to be
reformulated into an equivalent one, that could be processed in parallel.
Let us denote $B_{\new}\overset{df}{=}$\emph{splitsMap.get(pm)} and
\emph{lab} $\overset{df}{=}$ \emph{first(pm)}. The reason for $|B|>|B_{\new}|$
evaluation in (Listing: \ref{listing:bisimalg:marksplitcopyphase:marking}
Line: \ref{bisimalg:marksplitcopyphase:markerWorker:split_need_ident})
is to decide whether $\overline{B}_{\new}(\lab,B)=B$. Thus, the problem
is to decide whether $|B|>|B_{new}|$ can be reduced to a parallel
preparation of the $|\mathsf{S}|$-element array \emph{snew} containing
the state references of only those positions that correspond to the
ids of elements in $\overline{B}_{\new}(\lab,B)$. Then, the parallel
comparison of both \emph{snew} and $B$ needs to be performed.  Preparing
the \emph{snew} array involves $|\oute(B_{\new})|$ processors  ($|\mathsf{T}|$
at the worst case scenario). On the other hand, the parallel \emph{snew}
and \emph{B.vert}$ $ comparison requires $|\mathsf{S}|$ parallel
processors. Therefore, all the parallel\emph{ }evaluations of the
$|B|>|B_{\new}|$ statement involve at most $\max\{|\mathsf{T}|,|\mathsf{S}|^{2}\}$
concurrent processors.  The loop (Listing: \ref{listing:bisimalg:marksplitcopyphase:marking},
Lines: \ref{bisimalg:marksplitcopyphase:markerWorker:split_need_ident}
- \ref{bisimalg:marksplitcopyphase:markerWorker:states-mark}) iterating
through $v_{i}\in B_{\new}$\emph{.vert} can be executed in $O(1)$
with the help of as many parallel processors as the total size of
\emph{splitsMap}. Similarly, (Listing: \ref{listing:bisimalg:marksplitcopyphase:marking},
Line: \ref{bisimalg:marksplitcopyphase:markerWorker:foreach-splits-keys-end})
assignment also needs $O(1)$ of execution time. Summing up, the \emph{Marking()
}procedure achieves the parallel running time $O(1)$, and it requires
at most $\max\{|\mathsf{A}|\cdot|\mathsf{S}|^{2},|\mathsf{T}|\}$
of concurrent processors.

\paragraph{Splitting}

The first parallel iteration in \emph{Splitting()} goes through the
array \emph{M.mvert} (Listing: \ref{listing:bisimalg:marksplitcopyphase:Splitting}
Lines: \ref{bisimalg:marksplitcopyphase:splitterWorker:markerComp-start}
- \ref{bisimalg:marksplitcopyphase:splitterWorker:markerComp-end}).
It requires at most $|\mathsf{S}|$ processors to execute \emph{markersMap}
update in parallel. The values stored in \emph{markerMap} are also
the keys in \emph{subBlocksMap}. Therefore, they need to be uniquely
converted into single integers. Since there are, in fact, $k+1$ unsorted
tuples, their conversion to integers takes $O(\log(k+1))$ time of
$k+1$ processors (App. \ref{sec:parallel-index-creation-lema}).
Thus, because of $k+1\leq\beta$, the overall parallel running time
of the loop (Listing: \ref{listing:bisimalg:marksplitcopyphase:Splitting},
Lines: \ref{bisimalg:marksplitcopyphase:splitterWorker:markerComp-start}
- \ref{bisimalg:marksplitcopyphase:splitterWorker:markerComp-end})
is $O(\log\beta)$. The size of a single \emph{markersMap} array is
$O(|\mathsf{S}|$). The overall memory needed for a different \emph{markersMap}
is $O(|\mathsf{S}|^{2})$, whilst direct addressing of a \emph{subBlockMap}
requires $O(|\mathsf{A}|\cdot|\mathsf{S}|^{2\alpha})$. 

The second parallel loop (Listing: \ref{listing:bisimalg:marksplitcopyphase:Splitting},
Lines: \ref{bisimalg:marksplitcopyphase:splitterWorker:groupingByMark-start}
- \ref{bisimalg:marksplitcopyphase:splitterWorker:stateRemoval})
can be executed in $O(1)$ since both operations are reduced to accessing
the \emph{markersMap} and \emph{subBlocksMap} arrays. Similarly, removing
marked vertices (Listing: \ref{listing:bisimalg:marksplitcopyphase:Splitting},
Line: \ref{bisimalg:marksplitcopyphase:splitterWorker:stateRemoval})
also needs the time $O(1)$. Choosing the maximal block out of \emph{subBlockMap.values}
and $M$ requires $O(\log|\mathsf{S}|)$ operations (Listing: \ref{listing:bisimalg:marksplitcopyphase:Splitting},
Line: \ref{bisimalg:marksplitcopyphase:splitterWorker:choosingTheLargestPart})
and can be implemented in two parallel steps. In the first step, for
every block stored as a value in \emph{subBlockMap} and $M$, the
number of elements is computed, then, in the second step all the blocks
from \emph{subBlockMap} and $M$ are sorted, and the maximal block
is determined. Each of these two steps needs $O(\log|\mathsf{S}|)$
running time. The next few operations (Listing: \ref{listing:bisimalg:marksplitcopyphase:Splitting},
Lines: \ref{bisimalg:marksplitcopyphase:splitterWorker:else} - \ref{bisimalg:marksplitcopyphase:splitterWorker:emptyM-removal})
are either block comparisons or array operations, hence all of them
can be implemented in $O(1)$ using at most $|\mathsf{S}|$ processors.
The last parallel loop can be executed in $O(1)$ assuming that every
cell of \emph{subBlocksMap} has a separate processor assigned. Inside
this loop, the only instruction that needs to be executed in parallel
is update of the \emph{nextStateToBlockId} array. Since every $v\in B$
needs to be updated, the update requires at most $O(|\mathsf{S}|)$
processors. Summing up, the \emph{Splitting() }procedure achieves
the parallel running time $O(\log\beta)$, and it requires at most
$|\mathsf{A}|\cdot|\mathsf{S}|^{2\alpha}$ concurrent processors.

\subsubsection{Memory management\label{sub:The-demand-for-resources}}

High demand on shared memory and the number of processors comes from
the need to allocate large arrays. Each cell in such an array is visited
by a single processor. Usually, most of the cells are empty, therefore
the actual work performed by processors is much smaller than the number
of processors. Therefore, during the \emph{InitializationPhase(),}
and also during \emph{MarkSplitCopyPhase()} the number of processors
that do something more than finding that their cell is empty is at
most $|\mathsf{T}|$. The arrays may hold object references, but not
the object itself. Therefore, the single cell usually has a machine
word size. 

Many parts of the algorithm have $O(1)$ parallel running time. In
such places, knowing that the total time of the overall algorithm
is $O(|\mathsf{S}|\log\beta)$, demand on the hardware resources could
be optimized. For example, to handle the iteration through \emph{splitsMap}
in parallel, $\left\lceil \nicefrac{|\mathsf{A}|\cdot|\mathsf{S}|^{2\alpha}}{\log\beta}\right\rceil $
processors would be enough etc. 

In practice, the perfect parallel implementation may be a good option
if $\beta$ and $|\mathsf{A}|$ are small. For such a case, it may
be worthwhile implementing such a parallel model and use arrays as
structures, which guarantee the fully parallel access to the stored
data.

\section{Notes on the solution optimality\label{sec:Notes-on-the-solution-optimality}}

Finding a fast and efficient parallel algorithm solving \emph{RCPP}
is still a challenge for researchers. It has been shown that deciding
strong bisimilarity is \emph{P-Complete} \citep{Balcazar1992dbip},
which indicates that solutions should be sought in the class of problems
decidable in sequential time $|\mathsf{S}|^{O(1)}$ \citep{Greenlaw1995ltpc}.
Therefore, the reasonable running time for this class of problems
is $O(|\mathsf{S}|^{\epsilon})$ where $\epsilon$ is a small real
constant. In particular, the work \citep{Rajasekaran1998pafr} contains
conjecture that $\epsilon$ might be even smaller than $1$. In fact,
it is impossible and $\epsilon$ must be equal or greater than $1$.
Since, according to the best knowledge of the author, no one has so
far clearly stated this in the literature, a simple reasoning confirming
that $\epsilon\geq1$ is presented below. 
\begin{theorem}
The running time lower bound for any algorithm solving the bisimulation
problem is $O(|\mathsf{S}|)$\end{theorem}
\begin{proof}
Let \emph{LTS}$_{1}$$=(\mathsf{S},\mathsf{T},L)$ be a labelled transition
system, so that $\mathsf{S}=\{s_{1},\ldots,s_{2n}\}$, $\mathsf{T}=\{s_{i}\overset{a}{\rightarrow}s_{i+1}|i=1,\ldots,n-1,n+1,\ldots2n\}$$ $
and constant function $L=a$ (Fig. \ref{fig:lowerBoundExample}).
Thus, according to the definition of bisimulation (Def. \ref{bisim-def})
deciding whether $s_{1}\sim s_{n+1}$ requires answering the question
$s_{2}\sim s_{n+2}$ and so on, up to $s_{n}\sim s_{2n}$. Therefore,
an optimal (the fastest possible) decision algorithm answering the
question $s_{1}\sim s_{n+1}$ first needs to decide $s_{n}\sim s_{2n}$,
then in the next step $s_{n-1}\sim s_{2n-1}$, and finally in the
$n$'th step it is able to decide that $s_{1}\sim s_{n+1}$. None
of the steps may be skipped, since $s_{i}\sim s_{n+i}$ cannot be
resolved without knowing the result of $s_{i+1}\sim s_{n+i+1}$. Therefore,
there is no algorithm that solves the $s_{1}\sim s_{n+1}$ problem
in a smaller number of steps than $n$. The algorithm cannot be effectively
parallelized since there are no two different problems in the form
$s_{i}\sim s_{n+i}$ and $s_{j}\sim s_{n+j}$, where $1\leq i<j\leq n$,
that can be answered independently of each other. Therefore, the asymptotic
lower bound for the algorithm solving the problem of bisimilarity
between $s_{1}$ and $s_{n+1}$ in \emph{LTS}$_{1}$ is $O(\nicefrac{|\mathsf{S}|}{2})=O(|\mathsf{S}|)$.
In particular, this result indicates that the best algorithm solving
the bisimulation problem cannot run faster than $O(|\mathsf{S}|)$. 
\end{proof}
\begin{center}
\begin{figure}
\begin{centering}
\includegraphics[scale=1.3]{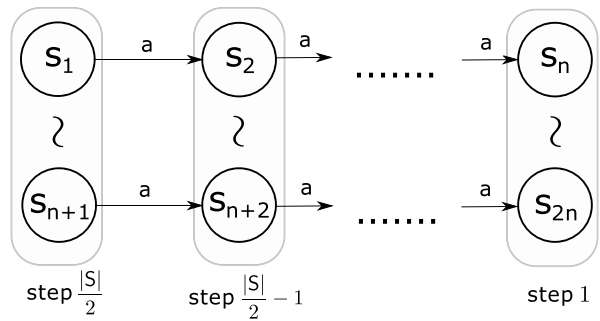}
\par\end{centering}

\caption{Problem $s_{1}\sim s_{n+1}$ is decidable in at least $\nicefrac{|\mathsf{S}|}{2}$
steps}
\label{fig:lowerBoundExample}
\end{figure}

\par\end{center}

The running time lower bound presented above is also valid for similar
problems. Since any algorithm solving \emph{RCPP }also solves bisimulation,
there is no algorithm that solves \emph{RCPP} asymptotically faster
than $O(|\mathsf{S}|)$. In particular, it is easy to observe that
for any $\epsilon<1$ there exists such (large enough) $\mathsf{S}$
for which $\nicefrac{|\mathsf{S}|}{2}>|\mathsf{S}|^{\epsilon}+O(1)$.
The graph induced by \emph{LTS}$_{1}$ is acyclic. This means that
the presented estimation is valid also when an input problem corresponds
to the bisimulation problem with an acyclic graph (well-founded set).
This observation makes the result \citep{Dovier2004aeaf} even more
important as it is asymptotically optimal for acyclic problems.  

Although it might seem that the further possibilities of improvement
in solving \emph{RCPP} are very limited, in fact, there is considerable
room for improvement. First of all, there is no answer whether \emph{RCPP}
could be solved sequentially faster than $O(|\mathsf{S}|\log|\mathsf{S}|)$
for any input data. In particular, it is an unknown algorithm running
in the linear time $O(|\mathsf{S}|)$. The parallel algorithms usually
suffer from the high demand on system resources, such as \emph{RAM}
and processors, or are too complex to be efficiently implemented in
practice. Thus, any attempt to reduce the demand on the resources
of existing algorithms, or to simplify the implementation, is valuable.

\section{Concurrent implementation\label{sub:Concurrent-implementation}}

The presented solution (Sec. \ref{sec:Bisimulation-Algorithm}) tries
to meet the demand for an efficient, concurrent and easy-to-implement
algorithm solving the \emph{RCPP} problem. It was initially developed
for use within the \emph{CCL} library - a formal notation library
designed mainly for modeling and executing behavior of concurrent
systems \citep{Kulakowski2014mili}. This prompted the author to look
for a simple to understand and easy to implement algorithm solving
the bisimulation problem. Therefore, the algorithm uses the basic
data structures, such as sets, queues or hash tables. The sole exception
is when the keys in a map are data collection itself. In such a case,
the hash code of such a collection must depend on the hash codes of
the elements. Sometimes it is the default behavior of the programming
language%
\footnote{see \emph{Java} platform, class \emph{java.util.AbstractSet}%
}. Often, however, it is reasonable to implement the hash code procedure
yourself. 

The range of keys of objects stored in the different structures varies
widely. Structures like the maps \emph{stateToBlockId, }or\emph{ nextStateToBlockId
}or the queues $\mathcal{S}$ and $\mathcal{M}$ hold at most $|\mathsf{S}|$
elements at the same time indexed from $1$ to $|\mathsf{S}|$ or
$|\mathsf{S}|\log|\mathsf{S}|$ respectively. Therefore, in order
to facilitate parallel processing, these structures could be implemented
in the form of arrays. On the other hand, there are such structures
as \emph{initPartition} or \emph{subBlocksMap.} The keys of objects
stored in them have the form of sorted sets. Hence, although it is
possible to determine the range of these keys, in practice it is better
to use concurrent objects \citep{Moir2005hods,Herlihy2008taom} to
implement these data structures. Of course, the degree of parallelism
is somewhat limited. Since, very often, the overall cost of method
call is high compared to the synchronization time, the expected slowdown
does not have to be significant %
\footnote{That is because many concurrent objects intensively use CAS (Compare
And Swap) \citep{Fraser2007dc} based synchronization. In such a case,
synchronization overhead is reduced to single processor instruction
calls. Moreover, an important threat to the operations' performance
can be the number of concurrently running threads see. e.g. \citep{Kirousis1992rmvi}.%
}. 

The presented algorithm is designed so that all the read and write
operations involving concurrent structures are grouped together. For
example, \emph{Marking()} gets the blocks from $\mathcal{S}$, processes
them, but modifies only the set of the marked states within the selected
block. Since \emph{Marking()} does not read the information about
the marked states, in fact marking operations (Listing: \ref{listing:bisimalg:marksplitcopyphase:marking},
Line: \ref{bisimalg:marksplitcopyphase:markerWorker:states-mark})
do not need to be synchronized with each other. Hence, the only synchronization
points between different \emph{Marking() }calls are limited to taking
blocks from $\mathcal{S}$ and adding blocks to $\mathcal{M}$. The
\emph{Splitting()} procedure also tries to follow the same design
scheme. It gets elements from $\mathcal{M}$ and adds the new elements
to $\mathcal{S}$, so that the interference between different \emph{Splitting()}
calls is minimal. Inside the \emph{Marking()} and \emph{Splitting(),}
as well as inside the \emph{InitializationPhase(),} procedures, there
are many places which can be processed concurrently. In practice,
there is no sense processing them all in separate threads. Some collections
have so few elements that the gain from parallel processing does not
compensate the time spent on launching of the new subtasks. Moreover,
often the degree of parallelism provided by the hardware platforms
is far from ideal. Thus, the possibly high granularity of the computing
tasks, does not always translate into the increase in the number of
threads actually executed in parallel. Therefore, when implementing
the algorithm, it is wise to limit the degree of parallelization of
the code.

The created experimental implementation%
\footnote{\href{http://www.kulakowski.org/concurrent-bisimulation}{www.kulakowski.org/concurrent-bisimulation}%
}  tries to find a trade-off between parallelization and effectiveness.
In particular, in the course of the experiments, it turned out that
the splitting according to the \emph{``process all but the largest
one''} strategy is so effective that the resulting blocks are usually
small. Hence, it turned out that in most cases the parallel iteration
over the elements of the newly separated block does not make sense. 

The test application was written in \emph{Java} 7 and has been tested
on an isolated test station Intel{\footnotesize\raisebox{1pt}{\textregistered}}~Core\texttrademark~i7-930
(4 cores, 8 threads, 2.8 GHz) processor with 16 GB of operating memory.
As input data for the tests, labelled transition systems from the
\emph{VLTS Benchmark Suite} (provided with the \emph{CADP} tool \citep{Garavel2011cadp})
were used, the largest of which had more than $10^{6}$ states and
more than $5\cdot10^{6}$ transitions. Every considered test case
has been computed five times. The first two runs were treated as a
warm-up, whilst the last three results have been averaged and taken
into account in the final result calculation. The algorithm speedup
$\nicefrac{T_{1}}{T_{t}}$, where $T_{1}$ means execution time of
a single-threaded application, and $T_{t}$ means $t$-thread application
is shown in (Fig. \ref{fig:speedup}).

\begin{figure}
\begin{centering}
\includegraphics[scale=0.5]{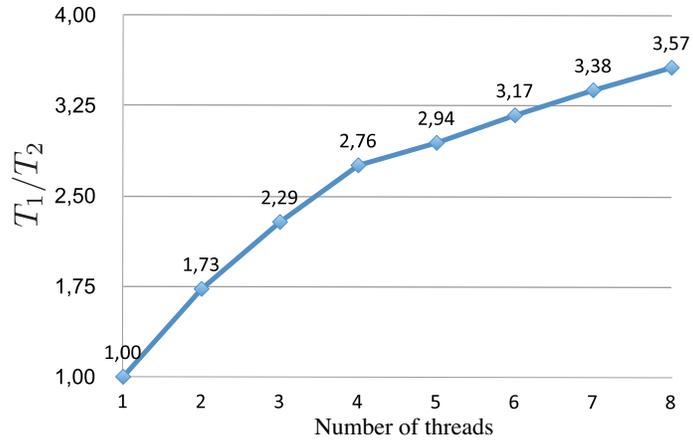}
\par\end{centering}

\caption{The algorithm speedup measured on the quad core Intel{\scriptsize \textregistered}~Core\texttrademark~i7-930
based machine}
\label{fig:speedup}
\end{figure}

The obtained results (maximal speedup $3,57$) seem satisfactory,
considering that the tests were run on a machine equipped with a four-core
processor. Theoretically, because of the eight computing threads,
the speedup could be higher. However, it should be noted that all
the tasks assigned to the working threads are memory intensive. Thus,
the actual limit for the speedup increase seems to be the number of
shared cache spaces of the hyper-threading architecture \citep{Intel2003ihtt,Athanasaki2007etpl}
rather than the availability of the \emph{CPU's} computing cores.

\section{Summary \label{sec:Summary}}

The article presents a new efficient concurrent algorithm for solving
the bisimulation problem. The main inspiration for the design of this
solution was the concept of the state signatures \citep{Blom2002adaf}
and the work of \emph{Paige} and \emph{Tarjan} \citep{Paige1984alta}.
The algorithm follows the principle - \emph{``process all but the
largest one'', }which is a modified version of the strategy introduced
by \emph{Hopcroft} \citep{Hopcroft1971anln}. The achieved expected
running time in the sequential and concurrent case is close to the
best alternative bisimulation algorithms \citep{Paige1984alta,Rajasekaran1998pafr}.
 The algorithm tries to be easy to implement. It uses hash maps,
sets and queues (and their concurrent counterparts) available in most
programming languages. Hence, in the opinion of the author, the algorithm
is likely to be useful for many professionals who want to improve
the performance of the existing solutions, or to implement new ones
from scratch.

\section*{Acknowledgement }

The author would like to thank Prof. Marcin Szpyrka, for taking the
time to read the manuscript. His remarks undoubtedly helped to improve
the final version of the article.  Special thanks are due to Ian
Corkill for his editorial help.

\bibliographystyle{plain}
\bibliography{papers_biblio_reviewed}

\appendix

\section{Creating an index of the k-element unordered set of numbers\label{sec:parallel-index-creation-lema}}

The algorithm shown below allows users to compute the index $i$ of
the $k$-element unordered set of numbers $\{i_{1},\ldots,i_{k}\}$,
where $i_{i}\in D_{1},\ldots,i_{k}\in D_{k}$ are finite intervals
in $\mathbb{N}$, in $O(\log k)$ time, using $k$ parallel processors.
The resulting index fits in the interval $[0,|D_{1}|\cdot\ldots\cdot|D_{k}|]$.
The input to the algorithm is $tbl$ - $k$-element array of integers.
It is assumed that the table may contain duplicates. The algorithm
consists of the following steps: 
\begin{enumerate}
\item parallel sort of $tbl$,
\item every processor with the number $i>0$ assigned to the $i$'th cell
of $tbl$ checks whether $tbl[i-1]=tbl[i]$, and if so puts the marker
$\infty\notin\mathbb{N}$ greater than any number in $I_{1},\ldots,I_{k}$
at the $i$'th position into the $tbl$ array,
\item parallel sort of $tbl$, 
\item every processor with the number $i\geq0$ checks whether $tbl[i]\neq\infty$
and $v.out[i+1]=\infty$. If so, it stores the index $i$ in the auxiliary
variable $ts$. 
\item every processor $i=0,\ldots,ts$ computes the value $tbl[i]\cdot|D_{i}|^{i}$
and stores it at the $i$'th position in the auxiliary table $tbl2$. 
\item the values stored in the $tbl2$ table are summed up into a single
integer.
\end{enumerate}
The first and the third step of the algorithm can be performed in
$O(\log k)$ with the help of $\frac{k}{\log k}\log\log k$ CRCW PRAM
processors \citep{Bhatt1991idpi}. Steps 2, 4 and 5 explicitly involve
at most $k$ processors, where every processor performs a simple action,
such as the comparison, multiplication or assignment of numbers. Thus
the parallel execution time of these steps is $O(1)$. The number
of summations in Step six could be performed in $O(\log k)$ using
$\frac{k}{\log k}$ EREW PRAM processors \citep{Ladner1980ppc}. 
\end{document}